\documentclass[letterpaper]{article}
\usepackage[margin=1in]{geometry}

\usepackage{amsmath}
\usepackage{amsfonts}
\usepackage{amsthm}

\usepackage{graphicx}
\usepackage{subcaption}
\usepackage{tikz}
\usepackage[sort,nocompress]{cite}

\newtheorem{definition}{Definition}
\newtheorem{theorem}{Theorem}
\newtheorem{lemma}{Lemma}
\newtheorem{proposition}{Proposition}
\newtheorem{corollary}{Corollary}
\newtheorem{remark}{Remark}
\newtheorem{assumption}{Assumption}

\usepackage{url}
\begin{document}

\title{Stability and Scalability of Blockchain Systems}
\author{Aditya Gopalan\thanks{University of Illinois at Urbana-Champaign, gopalan6@illinois.edu. Work done while affiliated with The University of Texas at Austin.} \and Abishek Sankararaman\thanks{University of California, Berkeley, abishek@berkeley.edu. Work done while affiliated with The University of Texas Austin.} \and
Anwar Walid\thanks{Nokia Bell Labs, anwar.walid@nokia-bell-labs.com} \and Sriram Vishwanath\thanks{The University of Texas at Austin, sriram@utexas.edu}}
\date{}


\maketitle
\begin{abstract}
The blockchain paradigm provides a mechanism for content dissemination and distributed consensus on Peer-to-Peer (P2P) networks.
While this paradigm has been widely adopted in industry, it has not been carefully analyzed in terms of its network scaling with respect to the number of peers.
Applications for blockchain systems, such as cryptocurrencies and IoT, require this form of network scaling.

In this paper, we propose a new stochastic network model for a blockchain system.
We identify a structural property called \emph{one-endedness}, which we show to be desirable in any blockchain system as it is directly related to distributed consensus among the peers.
We show that the stochastic stability of the network is sufficient for the one-endedness of a blockchain.
We further establish that our model belongs to a class of network models, called monotone separable models.
This allows us to establish upper and lower bounds on the stability region.
The bounds on stability depend on the connectivity of the P2P network through its conductance and allow us to analyze the scalability of blockchain systems on large P2P networks.
We verify our theoretical insights using both synthetic data and real data from the Bitcoin network.
\end{abstract}
\section{Introduction}
\label{sec:intro}
The blockchain paradigm, introduced in the Bitcoin whitepaper \cite{nakamoto2008bitcoin}, enables distributed consensus over a peer-to-peer network.
Each peer constantly mines new information called \emph{blocks}, which can consist of more fine-grained information called \emph{transactions}.
Thus, blocks in the network are created over time.
Each peer that creates (mines) a block also creates \emph{references} to one or more previously created blocks.
Peers also communicate blocks in order to synchronize their information sets; \emph{i.e.}, the sets of blocks and references the peers are aware of.

One of the main goals of a blockchain system is to enable consensus through distributed trust.
Trust is achieved by the references -- a peer only references a block for which they have verified the contents.
In order to achieve distributed consensus, all peers should trust the same blocks. 
If all peers trust a block, it is called \emph{confirmed}.
A natural performance requirement of a blockchain system is that the subset of blocks which are confirmed grows with time as blocks are created.
This, however, is not guaranteed as blocks are created over time at different locations in the network and then need to be disseminated.
Due to bandwidth limitations, communications on the network are not instantaneous and experience delays.
If blocks are created too quickly~\cite{bagaria2019prism, yang2019prism, eyal2016bitcoin}, the delays cause network congestion and can prevent a blockchain system from confirming blocks.
As blockchain technology matures and evolves~\cite{nakamoto2008bitcoin, buterin2013ethereum, popov2017iota, bagaria2019prism, eyal2016bitcoin, pass2017fruitchains, yang2019prism}, it is natural to study the scalability challenges that arise due to its adoption.
The scalability of distributed consensus protocols has been studied in various other contexts~\cite{sompolinsky2015secure, sanghavi2007gossiping, shah2009gossip, ugrinovskii2013conditions, dimakis2006geographic, ranganathan2001gossip, ioannidis2009optimal}.

The defining features that allow a blockchain system to confirm blocks are twofold: (1) the causality of block references, and (2) the dissemination of all blocks to all peers.
This paper presents a novel stochastic model of the core blockchain protocol which considers these key aspects to assess the impact of block creation rates, bandwidth limitations, and network topology on the performance of a blockchain system.
A blockchain system has two components -- a peer-to-peer network that disseminates blocks mined by a peer to others, and temporal dynamics, by which different peers mine blocks at different instances of time. 
The classical peer-to-peer models characterize the dissemination of blocks among peers through communication protocols, but do not capture the arrival of exogenous blocks. 
On the other hand, models in queueing networks, precisely characterize the temporal dynamics of block arrivals, but cannot capture dependencies between blocks imposed by the peer-to-peer gossip dynamics.

In this paper, we propose a new stochastic model that captures both the bandwidth-limited gossip-based dissemination of blocks among peers in a peer-to-peer network and exogenous block arrivals.
We model blocks as vertices and references as directed edges in a directed acyclic graph, as in \cite{papadis2018stochastic, bagaria2019prism}.
Formally, we model newly mined blocks as an arrival process such that each new block arrives to the network at the peer who mines it.
Upon a block arrival, a peer adds the new block to its local copy of the blockchain.
Peers also communicate over the network to synchronize their local copies of the blockchain based on a gossip-like protocol.
Due to bandwidth limitations, communicated blocks are subject to delays which depend on the instantaneous network congestion.
Precisely, each peer communicates blocks to neighboring peers at a given rate of communication, and communicates the oldest block not possessed by the other peer.
Thus, blocks experience a network delay, as they are disseminated in a First-Come First-Served (FCFS) basis by the peers.
It is therefore possible that a new block arrives to the system before the block(s) it references are confirmed.
In particular, not all peers necessarily have all previously arrived blocks upon the arrival of a new block.
Due to the aforementioned causality of block references, it is important to maintain the order in which blocks are disseminated across the peer-to-peer network so that a blockchain system can reliably confirm blocks.

In this paper, we study the problems of \emph{stability} and \emph{scalability} of our blockchain model.
Broadly, stability implies that there exists a positive block arrival rate at which blocks can be confirmed for a fixed peer-to-peer network.
The stability of blockchain systems ensures that an external observer can determine, in finite time, which blocks will eventually be confirmed.
In this paper, scalability implies that there exists a (fixed) positive block arrival rate at which blocks can be confirmed as the number of peers grows.
Scalability ensures that in our model, the performance of the blockchain system does not degrade as the number of peers participating in the system grows.
The problem of optimizing throughput by selecting system parameters and designing other protocols is thus dependent on first ensuring stability and scalability of distributed consensus on the underlying peer-to-peer network. 

In current implementations of blockchain systems, the block arrival rate is governed by algorithms such as Proof-of-Work and Proof-of-Stake \cite{nakamoto2008bitcoin, buterin2013ethereum}, but we abstract such algorithms to the more general notion of an arrival process.
This paper treats blocks as the atomic unit for blockchain systems.
However, in certain applications of blockchain systems, such as in cryptocurrencies, a block consists of multiple transactions, which are the atomic unit.
However, even in such systems (\textit{e.g.}~\cite{nakamoto2008bitcoin}), distributed consensus is achieved at the block level. 
We leave it to future work extend our model to include the dynamics at the transaction level.

\subsection{Contributions of this Paper}
This paper studies distributed consensus dynamics in blockchain networks and establishes the conditions under which such consensus can occur.
The contributions of this paper are threefold.

\noindent {\textbf{Asymptotic Structural Properties}} -- Motivated by the fact that blockchain systems require network resources (in the form of bandwidth) in order to confirm blocks in the ledger, we begin with a natural performance requirement.
As the bandwidth consumed by a blockchain system grows as time $t\to\infty$; thus the number of confirmed blocks should also grow.
To this end, we consider the evolution of a blockchain system in the limit as time $t\to\infty$ and show that if there are infinitely many confirmed blocks, the sub-graph of confirmed blocks exhibits the qualitative property of \emph{one-endedness}.
A precise definition of one-endedness is given in Section \ref{sec:definitions}.
Furthermore, we find that, the one-endedness of the limiting blockchain DAG is a sufficient condition for the existence of infinitely many confirmed blocks.

Armed with these results, we analyze two natural constructions of blockchain DAGs, which we refer to as the \emph{tree} and \emph{throughput-optimal} policies.
The tree policy is implemented in the Bitcoin and Ethereum blockchains, the two most widely-adopted blockchain systems \cite{nakamoto2008bitcoin, buterin2013ethereum}; the throughput-optimal policy is introduced in \cite{lewenberg2015inclusive}.
We show that, if the network is stable, then any blockchain constructed under these two policies has a one-ended limiting DAG. 
Thus, if the network is stable, these two constructions are able to confirm infinitely many blocks in the limit as time $t\to\infty$.

\noindent {\textbf{Stability and Scalability}} -- 
We compute bounds on the stability region of blockchain systems as a function of the block arrival rate, network bandwidth limitations, and network topology. 
Namely, we bound the maximum block arrival rate to the system such that a blockchain system using the tree or throughput-policy can confirm infinitely many blocks as time $t\to\infty$.
Precise definitions are given in Section \ref{sec:stability-scalability}.
Our analysis assumes that the input process is stationary, but not necessarily Poisson.
We find that $\mu$, the maximum block arrival rate to ensure stability, satisfies $\frac{\phi_H}{2\log N} \leq \mu \leq \inf_{S \subset H}\phi_H^{(S)}$, where $\phi_H$ is the conductance of the peer-to-peer network $H$, with $N$ peers, and $\phi_H^{(S)}$ is the conductance of the cut $S$.

Following the stability analysis, we use our bounds to assess the scalability of blockchain systems.
A sequence of peer-to-peer networks $(H_k)_{k\in\mathbb{N}}$ is scalable if there exists a positive block arrival rate $\lambda^*$ such that each network is stable with arrival rate $\lambda^*$.
We determine a necessary condition for scalability, which in turn provides a sufficient condition for the lack thereof.
We show as an example that sequences of peer-to-peer networks containing large stars not scalable.

\noindent {\textbf{Bitcoin System Evaluation using Real Data Traces}} -- Finally, we turn to numerical simulation to characterize quantitative measures of blockchain system network performance.
We prove that under the tree policy, stable blockchain systems confirm infinitely many blocks as time $t\to\infty$.
Our proof identifies a set of blocks called the \emph{distinguished path}, using which an external entity who is aware of the global network dynamics can determine, in finite time, which blocks will eventually be confirmed.
This determination relies on the network behavior of peer-to-peer dynamics in blockchain systems; from this result we are able to determine several network metrics for stable blockchain systems.
For small peer-to-peer networks, we simulate these performance metrics with respect to block arrival rates.

Using measurements of the Bitcoin peer-to-peer network taken by \cite{decker2013information} and \cite{gencer2018decentralization}, and traces of the Bitcoin peer-to-peer network taken by \cite{blockchair2020data} we compare the performance of our model with a simulated Poisson block arrival input and with an input of traces of the Bitcoin blockchain system.
We find that the assumption of Poisson block arrivals is a good approximation of the real process.

\subsection{Organization of this Paper}
In Section \ref{sec:model}, we present our stochastic network model and relevant definitions for our analysis.
In Section \ref{sec:structural-properties} we identify a structural relationship between confirmed blocks and provide a sufficient condition for the existence of infinitely many confirmed blocks as time $t\to\infty$.
In Section \ref{sec:one-ended-policies}, we show that stable blockchain systems using the tree and throughput-optimal policies confirm infinitely many blocks as time $t\to\infty$.
In Section \ref{sec:stability}, we derive bounds on the stability region for our model.
We also use our stability bounds to deduce the non-scalability of networks containing large stars and show that the per-peer block arrival rate decreases to $0$ with network size, even in a scalable sequence of networks.
In Section \ref{sec:interpretations}, we interpret our theoretical results and identify new network metrics to characterize stable blockchain dynamics.
In Section \ref{sec:simulations}, we conduct numerical experiments.
In Section \ref{sec:related_work}, we discuss related work.
We provide concluding remarks in Section \ref{sec:conclusion}. Proofs are deferred to appendices.

\section{System Model}
\label{sec:model}

Our model consists of a collection of peers on a peer-to-peer network.
Each peer adds blocks to the blockchain system in a process called \emph{mining}.
Newly mined blocks are added to the peers' individual copies of the global blockchain ledger.
Peers subsequently communicate newly mined blocks to other peers on the peer-to-peer network. 
Communications over the peer-to-peer network incur delays, which depend on the instantaneous network congestion. 
Each peer represents the instantaneous state of its copy of the blockchain as a DAG and updates its copy of the DAG according to both the communications received over the peer-to-peer network, as well as through block mining.
We describe this process more formally below.

\subsection{Stochastic Network Model}

\noindent {\bf Peer-to-Peer Network} - 
Our model consists of $N$ peers connected to each other by an undirected graph $H$. 
Each edge $(i, j)$ of $H$ represents a bi-directional communication link between peers $i$ and $j$.
Associated with each peer $p \in \{1, \ldots, N\}$, at each time $t \in \mathbb{R}_{+}$, is a DAG $G_p(t)$, whose vertex set is denoted by $B_p(t) \subset \mathbb{N} \cup \{0\}$ and edge set $E_p(t)$.
The DAGs $G_p(t),\ p \in \{1, \ldots, N\}$ represent the state of the blockchain from the perspective of peer $p$ at time $t$. 
The set $B_p(t)$ represents the set of blocks known to peer $p$ at time $t$ and the set $E_p(t)$ represents the aforementioned block references.

From henceforth and for clarity, $G_p(t)$, the blockchain graph at a peer $p$ at time $t$, and the union $G(t): = \bigcup_{0 \leq s \leq t} \bigcup_{p \in\{1, \ldots, N\}} G_p(s)$ are referred to as \emph{DAGs}. 
The vertices of any DAG are referred to as \emph{blocks} and the (directed) edges are referred to as \emph{references}.
Similarly, the graph $H$, which represents the communication structure among the $N$ peers is referred to as a \emph{network}.
The vertices of any network are referred to as \emph{peers} and the (undirected) edges are referred to as \emph{links}.


At time $0$, we assume that $G_p(0)$ is a single vertex indexed $0$, for all peers $p \in\{1, \ldots, N\}$. 
We denote by $B(t) = \bigcup_{p \in \{1, \ldots, N\}}B_p(t)$ and $E(t) = \bigcup_{p \in \{1, \ldots, N\}}E_p(t)$. 
The DAG $G(t)$ is the graph on the vertex set $B(t)$ with edge set $E(t)$.

\noindent {\bf Block Arrival and Reference Selection Process} - The DAGs $G_p(t)$ associated with the peers evolve with time as new blocks arrive to the system.
More precisely, the arriving blocks are indexed by the natural numbers  $\{1, 2, \ldots\}$.
Recall that at time $0$, all peers are in agreement about block $0$. 
Blocks arrive in continuous time (according to a stationary point process $A$ with intensity $\lambda$), with each block $i \in \mathbb{N}$ arriving at a (random) peer denoted by $p_i \in \{1, \ldots, N\}$. 
If a block $i$ arrives at peer $p$ at time $t$, we specify this event as peer $p$ \emph{mines} block $i$ at time $t$.
When a block indexed $i$ arrives at peer $p$, $p$ is instantly aware of the index of the newly arrived block.
In other words, the arriving block is instantly added to the block set of the DAG associated with peer $p$.
The outgoing references from block $i$ are chosen from among $B_{p}(t)\setminus \{i\}$ according to a fixed policy depending only on the DAG $G_{p}(t^{-})$, where $t^-$ is a moment in time infinitesimally before $t$. 
For each block $i \in \mathbb{N}$, we denote by $\mathcal{O}_i$ the set of outgoing neighbors of block $i$. 
If block $i$ is mined by peer $p$, at time $t$, $\mathcal{O}_i \subseteq B_{p}(t) \subseteq \{1, \ldots, i-1\}$. 
Notice that the set $\mathcal{O}_i$, is only chosen at the time of arrival by the peer to which block $i$ arrives to, and is fixed henceforth. 
We give examples of policies that select $\mathcal{O}_i$ in the sequel.

\noindent {\bf Communication Among Peers} - Associated with each peer $p \in \{1, \ldots, N\}$ is a marked point process $T_p$, for which each mark corresponds to another peer in $\{1, \ldots, N\} \setminus \{p\}$. 
At each epoch of $T_p$, peer $p$ contacts a peer $q$, given by the mark of the epoch. 
Instantly, peer $q$'s block set $B_{q}(t)$ is updated to include the lowest numbered block in $B_p(t) \setminus B_{q}(t)$ if this set is non-empty and the reference set $E_q(t)$ is also updated accordingly.
Observe that if peer $p$ communicates block $j \in \mathbb{N}$ to peer $q$ at time $t$, $\mathcal{O}_j \subseteq B_q(t)$. 
For otherwise, then one of the block in $\mathcal{O}_j$ would be communicated, as the communication policy sends the lowest numbered block and for every block $j$, peer $p$ and time $t$, $\mathcal{O}_j \subseteq \{1,\cdots,j-1\}$.

Note that the DAGs $G_p(t), p \in \{1, \ldots, N\}$ and $G(t)$ are \emph{random DAGs} as their growth is governed by a stochastic process.
These graphs are parameterized by the P2P network $H$ but we do not include this in our notation as the context will always be clear.

Observe that P2P network dynamics are a continuous time rumor-spreading process with exogenous arrivals \cite{ganeshnotes}.
Here, rumors represent blocks which are disseminated on the network.
For simplicity and without loss of generality, we assume that each peer $p$ has unit communication bandwidth, \textit{i.e.}, the process $T_p$ is rate $1$ block per second.
We relax this assumption in Remark \ref{remark-bandwidth}.
As a peer $p$ can communicate at most a single block at the epochs of $T_p$, the block dissemination is bandwidth-limited. 


\noindent {\bf Longest Chain Policies} - In this paper, we only consider the case where the outgoing edges of a block are chosen according to a class of \emph{deterministic} policies that we call \emph{Longest Chain Policies}.
This class of policies are such that for each arriving block $i \in \mathbb{N}$, which arrives at a random peer $p \in \{1, \ldots, N\}$, at time $t \in \mathbb{R}_{+}$, at least one of its outgoing edges connects to a vertex $j \in B_{p}(t)$, which is farthest away (in the sense of number of hops in $G_{p}(t)$) from block $0$. 
Formally, for each peer $p \in \{1,\ldots,N\}$ and time $t \in \mathbb{R}_{+}$, denote by the non-empty set
\begin{align}
\mathcal{L}_{p}(t):=\{ j \in B_{p}(t) : d(j,0) \geq d(j',0), \forall j' \in B_{p}(t) \},
\label{eqn:longest_chain_policy}
\end{align}
where $d(\cdot, 0)$ is the hop distance from $\cdot$ to 0 in $G_p(t)$.
The class of longest chain policies is such that for every block $i \in \mathbb{N}$ which arrives at peer $p$, at least one of its outgoing edges is in the set $\mathcal{L}_{p}(t)$. 
In other words, for every block $i \in \mathbb{N}$, that arrives at peer $p \in \{1, \ldots, N\}$, at time $t \in \mathbb{R}_{+}$, the set $\mathcal{O}_i \cap \mathcal{L}_{p}(t)$ is non-empty. 
This class of policies construct simple DAGs, \textit{i.e.}, for any two blocks $i > j \geq 0$, there is at most one directed edge from $i$ to $j$ in $G(t) := \bigcup_{p \in \{1, \ldots, N\}} G_p(t)$, for all $t \geq 0$.

In this paper, we consider the following two reference selection policies.
In both policies, we fix a block $i \in \mathbb{N}$, that arrives at a (random) peer $p \in \{1, \ldots, N\}$, at time $t \in \mathbb{R}_{+}$.
\begin{enumerate}
    \item \emph{Tree Policy} - $\mathcal{O}_i \subseteq \mathcal{L}_{p}(t),$ such that $|\mathcal{O}_i| = 1$.
    Every block has exactly one outgoing reference, chosen according to a deterministic rule from the set $\mathcal{L}_{p}(t)$. 
    We assume without loss of generality that each block $i$ has an outgoing reference to the least indexed block in $\mathcal{L}_{p}(t)$ in the event that $|\mathcal{L}_{p}(t)| > 1$.
    \item \emph{Throughput Optimal Policy} - \\
    $\mathcal{O}_i = \{b \in B_p(t) : b \text{ is a leaf in } G_p(t)\}.$
    Every block connects to all leaves in $G_p(t^-)$.
    We explain after Corollary \ref{cor:throughput-all-confirmed} why this policy is called throughput-optimal.
    
\end{enumerate}

In this paper, we only analyze blockchain systems that use the tree and throughput-optimal policies.
Bitcoin \cite{nakamoto2008bitcoin} and Ethereum \cite{buterin2013ethereum}, the two most widely adopted blockchain implementations, both use the tree policy, and the throughput-optimal policy is studied in \cite{lewenberg2015inclusive}.

\noindent {\bf Example Blockchain Realization} - 
See Figure \ref{fig:realization} for an example realization of the arrival and transmission processes on a peer-to-peer network with 2 peers, $p$ and $q$.
Both peers use the tree policy.
In the example, the arrival process $A$ has points at times $1.1, 2.4, 4.0$, and $6.2$, with marks $p, p, q$, and $q$, respectively.
The transmission process $T_{p}$ occurs at times $2.6, 5.2$ and the transmission process $T_{q}$ occurs at times $5.8, 6.9$.
The figure depicts the DAG $G(t)$ throughout the duration of these point processes and enumerates the sets $B_p(t), B_q(t)$.
Subfigures (b), (c), (d), (e), (f), (g), (h) capture the system in increasing time.

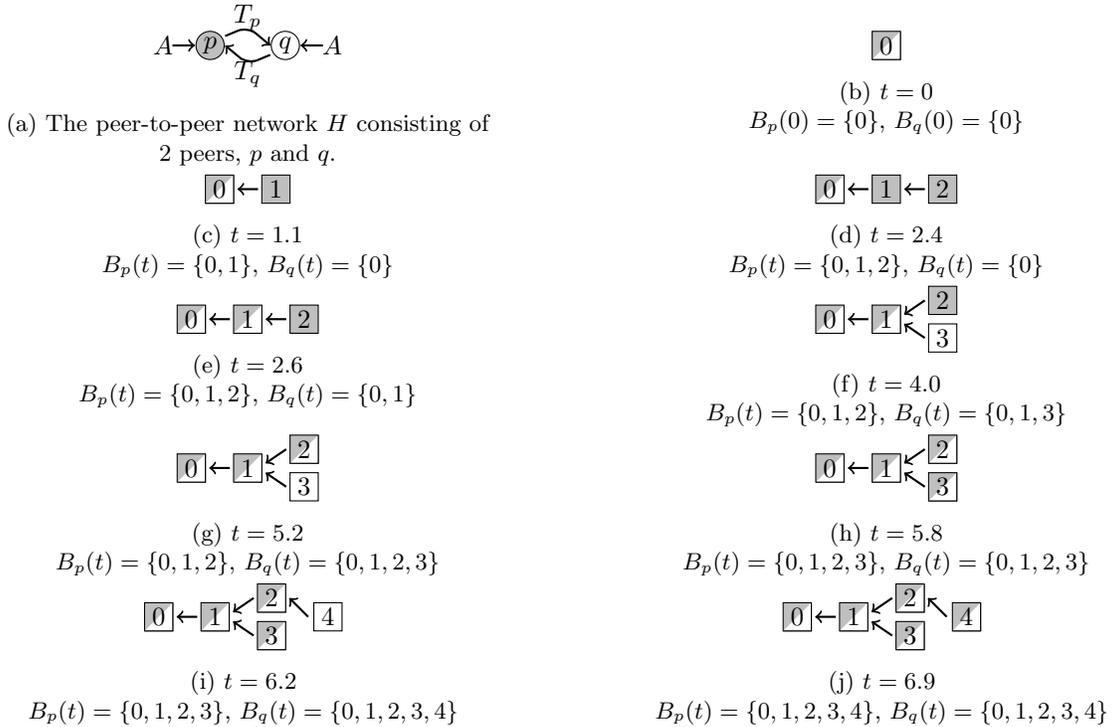
\begin{figure}[htp]
    \centering
    \captionsetup{justification=centering}
    \begin{subfigure}[!bt]{0.4\columnwidth}
    {\centering
    \begin{tikzpicture}[scale=0.5]
        \draw [fill=lightgray] (-1, 0) circle [radius=0.375];
        \node at (-1, 0) {$p$};
        \draw [fill=white] (1, 0) circle [radius=0.375];
        \node at (1, 0) {$q$};
        
        \draw [->, rounded corners, thick] (-0.6, 0.25) -- (0, 0.5) -- (0.6, 0);
        \node at (0, 0.75) {$T_{p}$};
        \draw [->, rounded corners, thick] (0.6, -0.25) -- (0, -0.5) -- (-0.6, 0);
        \node at (0, -0.75) {$T_{q}$};
            
        \draw [->, thick] (-2, 0) -- (-1.4, 0);
        \node at (-2.25, 0) {$A$};
        \draw [->, thick] (2, 0) -- (1.4, 0);
        \node at (2.25, 0) {$A$};
        \end{tikzpicture}
        \par}
    \caption{The peer-to-peer network $H$ consisting of 2 peers, $p$ and $q$.}
\end{subfigure}
\hspace{0.1\columnwidth}
\begin{subfigure}[!bt]{0.4\columnwidth}
    {\centering
    \begin{tikzpicture}[scale=0.5]
        \draw [fill=lightgray] (0, 0) rectangle (0.75, 0.75);
        \filldraw [fill=white] (0, 0) -- (0.75, 0) -- (0.75, 0.75);
        \node at (0.375, 0.375) {0};
    \end{tikzpicture}
    \par}
    \caption{$t = 0$
    \\
    $B_p(0) = \{0\}$, $B_q(0) = \{0\}$}
\end{subfigure}
\\
\begin{subfigure}[!bt]{0.4\columnwidth}
    {\centering
    \begin{tikzpicture}[scale=0.5]
        \draw [fill=lightgray] (0, 0) rectangle (0.75, 0.75);
        \filldraw [fill=white] (0, 0) -- (0.75, 0) -- (0.75, 0.75);
        \node at (0.375, 0.375) {0};
            
        \draw [fill=lightgray] (1.5, 0) rectangle (2.25, 0.75);
        \node at (1.875, 0.375) {1};
            
        \draw [->, thick] (1.4, 0.375) -- (0.85, 0.375);
    \end{tikzpicture}
    \par}
    \caption{$t = 1.1$
    \\
    $B_p(t) = \{0, 1\}$, $B_q(t) = \{0\}$}
\end{subfigure}
\hspace{0.1\columnwidth}
\begin{subfigure}[!bt]{0.4\columnwidth}
    {\centering
    \begin{tikzpicture}[scale=0.5]
        \draw [fill=lightgray] (0, 0) rectangle (0.75, 0.75);
        \filldraw [fill=white] (0, 0) -- (0.75, 0) -- (0.75, 0.75);
        \node at (0.375, 0.375) {0};
            
        \draw [fill=lightgray] (1.5, 0) rectangle (2.25, 0.75);
        \node at (1.875, 0.375) {1};
            
        \draw [fill=lightgray] (3, 0) rectangle (3.75, 0.75);
        \node at (3.375, 0.375) {2};
            
        \draw [->, thick] (1.4, 0.375) -- (0.85, 0.375);
        \draw [->, thick] (2.9, 0.375) -- (2.35, 0.375);
    \end{tikzpicture}
    \par}
    \caption{$t = 2.4$
    \\
    $B_p(t) = \{0, 1, 2\}$, $B_q(t) = \{0\}$}
\end{subfigure}   
\\
\begin{subfigure}[!bt]{0.4\columnwidth}
    {\centering
    \begin{tikzpicture}[scale=0.5]
        \draw [fill=lightgray] (0, 0) rectangle (0.75, 0.75);
        \filldraw [fill=white] (0, 0) -- (0.75, 0) -- (0.75, 0.75);
        \node at (0.375, 0.375) {0};
            
        \draw [fill=lightgray] (1.5, 0) rectangle (2.25, 0.75);
        \filldraw [fill=white] (1.5, 0) -- (2.25, 0) -- (2.25, 0.75);
        \node at (1.875, 0.375) {1};
            
        \draw [fill=lightgray] (3, 0) rectangle (3.75, 0.75);
        \node at (3.375, 0.375) {2};
            
        \draw [->, thick] (1.4, 0.375) -- (0.85, 0.375);
        \draw [->, thick] (2.9, 0.375) -- (2.35, 0.375);
    \end{tikzpicture}
    \par}
    \caption{$t = 2.6$
    \\
    $B_p(t) = \{0, 1, 2\}$, $B_q(t) = \{0, 1\}$}
\end{subfigure}
\hspace{0.1\columnwidth}
\begin{subfigure}[!bt]{0.4\columnwidth}
    {\centering
    \begin{tikzpicture}[scale=0.5]
        \draw [fill=lightgray] (0, 0) rectangle (0.75, 0.75);
        \filldraw [fill=white] (0, 0) -- (0.75, 0) -- (0.75, 0.75);
        \node at (0.375, 0.375) {0};
            
        \draw [fill=lightgray] (1.5, 0) rectangle (2.25, 0.75);
        \filldraw [fill=white] (1.5, 0) -- (2.25, 0) -- (2.25, 0.75);
        \node at (1.875, 0.375) {1};
            
        \draw [fill=lightgray] (3, 0.5) rectangle (3.75, 1.25);
        \node at (3.375, 0.875) {2};
            
        \draw [fill=white] (3, 0.25) rectangle (3.75, -0.5);
        \node at (3.375, -0.125) {3};
            
        \draw [->, thick] (1.4, 0.375) -- (0.85, 0.375);
        \draw [->, thick] (2.9, 0.875) -- (2.35, 0.5);
        \draw [->, thick] (2.9, -0.125) -- (2.35, 0.25);
    \end{tikzpicture}
    \par}
    \caption{$t = 4.0$
    \\
    $B_p(t) = \{0, 1, 2\}$, $B_q(t) = \{0, 1, 3\}$}
\end{subfigure}   
\\
\begin{subfigure}[!bt]{0.4\columnwidth}
    {\centering
    \begin{tikzpicture}[scale=0.5]
        \draw [fill=lightgray] (0, 0) rectangle (0.75, 0.75);
        \filldraw [fill=white] (0, 0) -- (0.75, 0) -- (0.75, 0.75);
        \node at (0.375, 0.375) {0};
            
        \draw [fill=lightgray] (1.5, 0) rectangle (2.25, 0.75);
        \filldraw [fill=white] (1.5, 0) -- (2.25, 0) -- (2.25, 0.75);
        \node at (1.875, 0.375) {1};
            
        \draw [fill=lightgray] (3, 0.5) rectangle (3.75, 1.25);
        \filldraw [fill=white] (3, 0.5) -- (3.75, 0.5) -- (3.75, 1.25);
        \node at (3.375, 0.875) {2};
            
        \draw [fill=white] (3, 0.25) rectangle (3.75, -0.5);
        \node at (3.375, -0.125) {3};
            
        \draw [->, thick] (1.4, 0.375) -- (0.85, 0.375);
        \draw [->, thick] (2.9, 0.875) -- (2.35, 0.5);
        \draw [->, thick] (2.9, -0.125) -- (2.35, 0.25);
    \end{tikzpicture}
    \par}
    \caption{$t = 5.2$
    \\
    $B_p(t) = \{0, 1, 2\}$, $B_q(t) = \{0, 1, 2, 3\}$}
\end{subfigure}
\hspace{0.1\columnwidth}
\begin{subfigure}[!bt]{0.4\columnwidth}
    {\centering
    \begin{tikzpicture}[scale=0.5]
        \draw [fill=lightgray] (0, 0) rectangle (0.75, 0.75);
        \filldraw [fill=white] (0, 0) -- (0.75, 0) -- (0.75, 0.75);
        \node at (0.375, 0.375) {0};
            
        \draw [fill=lightgray] (1.5, 0) rectangle (2.25, 0.75);
        \filldraw [fill=white] (1.5, 0) -- (2.25, 0) -- (2.25, 0.75);
        \node at (1.875, 0.375) {1};
            
        \draw [fill=lightgray] (3, 0.5) rectangle (3.75, 1.25);
        \filldraw [fill=white] (3, 0.5) -- (3.75, 0.5) -- (3.75, 1.25);
        \node at (3.375, 0.875) {2};
            
        \draw [fill=lightgray] (3, 0.25) rectangle (3.75, -0.5);
        \filldraw [fill=white] (3, -0.5) -- (3.75, -0.5) -- (3.75, 0.25);
        \node at (3.375, -0.125) {3};
            
        \draw [->, thick] (1.4, 0.375) -- (0.85, 0.375);
        \draw [->, thick] (2.9, 0.875) -- (2.35, 0.5);
        \draw [->, thick] (2.9, -0.125) -- (2.35, 0.25);
    \end{tikzpicture}
    \par}
    \caption{$t = 5.8$
    \\
    $B_p(t) = \{0, 1, 2, 3\}$, $B_q(t) = \{0, 1, 2, 3\}$}
\end{subfigure}
\\
\begin{subfigure}[!bt]{0.4\columnwidth}
    {\centering
    \begin{tikzpicture}[scale=0.5]
        \draw [fill=lightgray] (0, 0) rectangle (0.75, 0.75);
        \filldraw [fill=white] (0, 0) -- (0.75, 0) -- (0.75, 0.75);
        \node at (0.375, 0.375) {0};
            
        \draw [fill=lightgray] (1.5, 0) rectangle (2.25, 0.75);
        \filldraw [fill=white] (1.5, 0) -- (2.25, 0) -- (2.25, 0.75);
        \node at (1.875, 0.375) {1};
            
        \draw [fill=lightgray] (3, 0.5) rectangle (3.75, 1.25);
        \filldraw [fill=white] (3, 0.5) -- (3.75, 0.5) -- (3.75, 1.25);
        \node at (3.375, 0.875) {2};
            
        \draw [fill=lightgray] (3, 0.25) rectangle (3.75, -0.5);
        \filldraw [fill=white] (3, -0.5) -- (3.75, -0.5) -- (3.75, 0.25);
        \node at (3.375, -0.125) {3};
        
        \draw [fill=white] (4.5, 0) rectangle (5.25, 0.75);
        \node at (4.875, 0.375) {4};
            
        \draw [->, thick] (1.4, 0.375) -- (0.85, 0.375);
        \draw [->, thick] (2.9, 0.875) -- (2.35, 0.5);
        \draw [->, thick] (2.9, -0.125) -- (2.35, 0.25);
        \draw [->, thick] (4.35, 0.375) -- (3.85, 0.875);
    \end{tikzpicture}
    \par}
    \caption{$t = 6.2$
    \\
    $B_p(t) = \{0, 1, 2, 3\}$, $B_q(t) = \{0, 1, 2, 3, 4\}$}
\end{subfigure}
\hspace{0.1\columnwidth}
\begin{subfigure}[!bt]{0.4\columnwidth}
    {\centering
    \begin{tikzpicture}[scale=0.5]
        \draw [fill=lightgray] (0, 0) rectangle (0.75, 0.75);
        \filldraw [fill=white] (0, 0) -- (0.75, 0) -- (0.75, 0.75);
        \node at (0.375, 0.375) {0};
            
        \draw [fill=lightgray] (1.5, 0) rectangle (2.25, 0.75);
        \filldraw [fill=white] (1.5, 0) -- (2.25, 0) -- (2.25, 0.75);
        \node at (1.875, 0.375) {1};
            
        \draw [fill=lightgray] (3, 0.5) rectangle (3.75, 1.25);
        \filldraw [fill=white] (3, 0.5) -- (3.75, 0.5) -- (3.75, 1.25);
        \node at (3.375, 0.875) {2};
            
        \draw [fill=lightgray] (3, 0.25) rectangle (3.75, -0.5);
        \filldraw [fill=white] (3, -0.5) -- (3.75, -0.5) -- (3.75, 0.25);
        \node at (3.375, -0.125) {3};
        
        \draw [fill=lightgray] (4.5, 0) rectangle (5.25, 0.75);
        \filldraw [fill=white] (4.5, 0) -- (5.25, 0) -- (5.25, 0.75);
        \node at (4.875, 0.375) {4};
            
        \draw [->, thick] (1.4, 0.375) -- (0.85, 0.375);
        \draw [->, thick] (2.9, 0.875) -- (2.35, 0.5);
        \draw [->, thick] (2.9, -0.125) -- (2.35, 0.25);
        \draw [->, thick] (4.35, 0.375) -- (3.85, 0.875);
    \end{tikzpicture}
    \par}
    \caption{$t = 6.9$
    \\
    $B_p(t) = \{0, 1, 2, 3, 4\}$, $B_q(t) = \{0, 1, 2, 3, 4\}$}
\end{subfigure}
    \captionsetup{justification=justified, singlelinecheck=false}
    \caption{A sample realization of $G(t)$ wherein blocks are added according to the tree policy.
    Here, gray blocks represent blocks known to Peer $p$ and white blocks are known to Peer $q$.
    Blocks shaded both gray and white are known to both peers $p$ and $q$.}
    \label{fig:realization}
\end{figure}


\begin{remark}
A typical assumption made in the blockchain literature (see, \textit{e.g.}, \cite{nakamoto2008bitcoin, bagaria2019prism, papadis2018stochastic}) is that the block arrival process is Poisson.
Under this assumption, our system dynamics are Markovian as all arrivals and communications occur with independent and exponentially distributed increments.
In this paper, we relax the Poisson arrival assumption, and treat the block arrival process as a general stationary and ergodic process.

Decker and Wattenhofer \cite{decker2013information} find that information propagation in blockchain systems resembles gossip-based P2P information dissemination.
Our communication model is motivated by their findings and our assumption of Poisson intercommunication times among peers is consistent with the implementation of the Bitcoin blockchain \cite{fanti2018dandelion++}.
\end{remark}

\subsection{Stablility and Scalability}
\label{sec:stability-scalability}
We now define stability and scalability for blockchain-like systems, which are motivated from distributed agreement among peers.

\begin{definition}
  A blockchain system is \emph{consistent} at time $t$ if $B_p(t) = B(t)$ for all peers $p \in \{1, \ldots, N\}$.
  Such a time $t$ is called a \emph{time of consistency}. 
  In other words, the system is consistent at time $t$ if all peers have identical block sets.
  \label{def:consistent}
\end{definition}

In Figure \ref{fig:realization}, subfigures (b), (h), and (j) are times of consistency.
We discuss consistency and its relation to analogous concepts in the queueing theory literature in Sections \ref{sec:interpretations} and \ref{sec:simulations}.

\begin{definition}
  A blockchain system is \emph{stable} (or \emph{recurrent}) for a given block arrival rate $\lambda$ if, almost surely, there exists an infinite sequence $(C_k)_{k \in \mathbb{N}}$ of times such that the system is consistent, and $G(C_j) \neq G(C_k)$ if $j \neq k$.
  \label{def:stability}
\end{definition}

In other words, the blockchain system is stable if there are infinitely many times of consistency with at least one block arrival between subsequent times of consistency. 
We now define scalability.

\begin{definition}
  Consider a sequence $(H_k)_{k \in \mathbb{N}}$ of P2P communication networks, where the number of peers in $H_k$ increases monotonically as $k$ grows.
  A blockchain system is \emph{scalable on $(H_k)_k$}  
  if there exists a non-zero block arrival rate $\lambda^* > 0$ such that, for every $k \in \mathbb{N}$, the blockchain system is stable with rate $\lambda^*$ on the P2P network $H_k$.
  If there is no such non zero block arrival rate $\lambda^*$, the system with P2P networks $(H_k)_{k \in \mathbb{N}}$ is {\emph{non-scalable}}.
  \label{def:scalability}
\end{definition}
 
Stability is a property of the blockchain system for a fixed P2P network, in the limit as time goes to infinity. 
Notice that stability is an asymptotic concept and can only be assessed in the limit as time goes to infinity. 
In addition, stability is a property only of the block communication dynamics on a P2P network and the block arrival rate and is not dependent on the reference selection policy by which individual peers add blocks. This follows as stability only requires that all peers be aware of the same set of blocks infinitely often, and the communication among peers is governed only by block indices.

Scalability, on the other hand, is a \emph{double limit}, where for a fixed P2P system, we take time to infinity, and then subsequently, take the size of the P2P system to infinity. 
Thus, the definition of scalability is with respect to a particular sequence of P2P networks, whose vertex sets grow to infinity.
Precisely, we deem an infinite sequence of P2P networks scalable if there exists a non-zero arrival rate such that a blockchain system on every P2P network of the sequence is stable for that rate.



\subsubsection{Connections to Peer-to-Peer and Queueing Networks}

As the P2P network is bandwidth-limited, stability and scalability are not guaranteed \textit{a priori}.  
To see this, fix the number of peers $k \in \mathbb{N}$ and the bandwidth at all peers to be $1$, \textit{e.g.}, the intensity of the communication processes $(T_p)_{p=1}^k$ are all equal to $1$.
In such a setting, the \emph{total bandwidth} in the network increases with the number of peers $k$, as each peer has unit outgoing bandwidth.
However, as each block that arrives to the system must be communicated to all other peers to ensure stability, increasing the number of peers in the system may increase the dissemination time of any given block.
Our definitions of stability and scalability, capture this trade-off and facilitate the analysis of large-scale \emph{bandwidth-limited} blockchain systems.


The definitions of stability and scalability resemble those studied in classical queueing networks \cite{kelly2011reversibility}. 
Indeed, one can view our process, as consisting of blocks (or ``customers,'' in the standard queueing literature), arriving into a ``queue'' at rate $\lambda$. 
Each block (customer) ``leaves'' once it has been disseminated to all the peers in the peer-to-peer network. 
However, the rate of service given to any block (customer), is not a function of the queue state alone, as in either the First-Come First-Served (FCFS) or Generalized Processor Sharing disciplines. 
Instead, the rate of service depends on the internal blockchain states of the peers in the peer-to-peer network and the associated communication processes among peers.
Thus, a direct analysis of stability cannot be conducted by coupling our model to any equivalent queueing network model.
Despite the difference between our model and traditional queueing networks, we use some of the technical ideas developed to study networks of queues (\cite{baccelli1995saturation}), such as the monotone coupling and saturation rules, to analyze our system.




\subsection{Preliminaries on Infinite DAGs}
\label{sec:definitions}
In subsection, we provide definitions and key properties of infinite DAGs.
These will be used to describe the blockchain DAG later in the paper, where the vertices will be blocks and the (directed) edges will be references.


\begin{definition}
  A \emph{maximal path} from a vertex $i$ in a DAG is any path beginning at $i$ and ending at some vertex $j$ such that no edges originate at $j$.
\end{definition}
Note that maximal paths exist for all vertices in finite DAGs; but in general this need not be true for infinite DAGs.
In general, a vertex may have more than one maximal path. 

Halin \cite{halin1964unendliche} introduces the concept of \emph{ends} in infinite graphs, which we re-state for infinite DAGs below.
\begin{definition}
  An \emph{infinite path} in a DAG $G(V, E)$ is a sequence of vertices $(v_k)_{k \in \mathbb{N}}$ such that either the directed edge $(v_i, v_{i+1}) \in E$ for all $i \in \mathbb{N}$, or the directed edge $(v_{i+1}, v_i) \in E$ for all $i \in \mathbb{N}$.
\end{definition}
Note that any finite or infinite path in a DAG does not contain any repeated vertices and has either a first or a last vertex $v_1 \in V$.

\begin{definition}[Halin \cite{halin1964unendliche}]
  Two infinite paths $p_1, p_2$ in an infinite DAG are \emph{equivalent} if there exists a third infinite path $p_3$ with $|p_3 \cap p_1| = |p_3 \cap p_2| = \infty$, where the intersection is over vertices.
  \label{def:path_equivalence}
\end{definition}

The equivalence relation in Definition \ref{def:path_equivalence} partitions the set of all infinite paths in a DAG into equivalence classes called \emph{ends}.
\begin{definition}[Halin \cite{halin1964unendliche}]
  For any $n \in \mathbb{N} \cup \{0, \infty\}$, an infinite DAG is \emph{n-ended} if it has $n$ ends, \textit{i.e.}, there are exactly $n$ different equivalence classes of infinite paths. 
  An infinite DAG is \emph{one-ended} if all of its infinite paths are equivalent. If there are no infinite paths in DAG, then it has $0$ ends.
  \label{def:one-ended}
\end{definition}
Note that under Definition \ref{def:one-ended}, all finite DAGs have 0 ends.

We give the following examples to illustrate the number of ends in various graphs.
All of these examples are on the vertex set $\mathbb{N}$.
\begin{enumerate}
    \item There is an edge from vertex $i$ to vertex $1$ for all $i > 1$, and no other edges.
    As there are no infinite paths, this graph is $0$ ends.
    \item There is an edge from vertex $i$ to vertex $i + 1$ for all $i \in \mathbb{N}$.
    Here, any infinite path is necessarily of the form $j, j+1, \ldots$ for some $j \in \mathbb{N}$.
    For any two infinite paths beginning at $j, k \in \mathbb{N}$, the path $\max(j, k), \max(j, k)+1, \ldots$ intersects both paths infinitely often; hence this graph is one-ended.
    \item There is an edge from vertex $1$ to vertex $2$.
    In addition, there is an edge from vertex $i$ to vertex $i + 2$.
    The following infinite paths are not equivalent: $1, 3, 5, \ldots$ and $2, 4, 6, \ldots$.
    Thus this graph is $2$-ended.
\end{enumerate}
Further work on ends in random graphs can be found in \cite{aldous2007processes, baccelli2018eternal, baccelli2018renewal}.

\begin{definition}
  A DAG is \emph{locally finite} if each vertex has finite in-degree and finite out-degree.
  \label{def:locally-finite}
\end{definition}

Note that any finite DAG is also locally finite; however an infinite DAG need not be locally finite.

\begin{proposition}
  Suppose a block $i \in \mathbb{N}$ arrives to the system at time $t_i$ and chooses its neighbors using a fixed edge selection policy.
  Then every maximal path from $i$ in $G(t)$ ends at block 0 for all $t \geq t_i$.
  \label{prop:maximal_path}
\end{proposition}
The proof is given in Appendix \ref{prf:maximal_path}.
   

\begin{definition}
  Under the tree policy, we define the \emph{distinguished path} of any finite DAG $G$ to be the longest maximal path. 
  If the longest maximal path is not unique, then the distinguished path is the longest maximal path beginning at the vertex having the least index.
  \label{def:distinguished_path}
\end{definition}
Under the tree policy, the distinguished path in $G_p(t)$, for any peer $p$ and time $t \geq 0$, is the longest maximal path beginning from the vertex $i := \inf \{b \in \mathcal{L}_p(t)\}$.

\section{Asymptotic Properties of $G(t)$}
\label{sec:structural-properties}
Recall that the goal of a blockchain system is such that a set of $N$ anonymous peers can agree on an ordered set of events without asking each other for local state information.
A natural requirement is that the $N$ peers should be able to agree on an infinite set of events if given an infinite time horizon.
In this section, we take the limit as $t \to \infty$ in order to determine precisely a set of blocks which are agreed upon by all peers -- we call these blocks \emph{confirmed}.
We show in Lemma \ref{lem:infinitely-confirmed} that when there are infinitely many confirmed blocks, the subgraph of confirmed blocks is one-ended.
The main result in this section, stated in Lemma \ref{lem:one-ended-confirmed}, shows that one-endedness of the blockchain DAG in the limit as $t \to \infty$ guarantees the existence of infinitely many confirmed blocks.

\subsection{Limiting Blockchain DAG}
Our analyses in this paper involve studying an infinite DAG which is constructed by iteratively adding a single vertex to a finite DAG.
We denote by the \emph{limiting blockchain DAG} $G(\infty) := \bigcup_{t \geq 0} G(t) = \bigcup_{t \geq 0} \bigcup_{p=1}^N G_p(t)$. 
The DAG $G(\infty)$ is infinite, as its vertex set consists of all blocks, i.e., the vertex set of $G(\infty)$ is $\mathbb{N}$. 
Observe from the definition of $G(t)$, that the map $t \rightarrow G(t)$ is monotone non-decreasing. 
This follows as the outgoing edges for a block are never changed after it is created at the instant of block arrival, and blocks (vertices) are never removed. 
Notice that for all times $0 \leq t < \infty$, $G(t)$ is a almost surely finite, as almost surely, in any finite interval of time, only finitely many blocks arrive to the system. 
Notice that almost surely, for all times $0 \leq t < \infty$, $G(t)$ is finite, as only finitely many blocks arrive to the system before any finite $t$.
Thus, $G(\infty) := \bigcup_{t \geq 0} G(t)$ is well-defined as it can almost surely be expressed as a countable union of DAGs. In Appendix \ref{app:technical-dag}, we give additional technical details and establish that $G(\infty)$ as an appropriate \emph{limit} of a sequence of finite graphs, and thus we call it the \emph{limiting blockchain DAG}.


\subsection{Confirmed Blocks in $G(\infty)$}

In this subsection give a precise definition of a \emph{confirmed block}.
In the sequel we identify the relationship between confirmed blocks in $G(\infty)$ and identify a sufficient condition such that infinitely many confirmed blocks exist.

\begin{definition}
  A block $b$ in $G(\infty)$ is \emph{confirmed} if all but finitely many blocks of index greater than $b$ have a path to $b$ in $G(\infty)$.
  \label{def:confirmed}
\end{definition}

Observe that the Definition \ref{def:confirmed} is an asymptotic property; namely it can only be verified in the limit as $t\to\infty$ and not at any finite time.

\begin{definition}
  A peer $p$ \emph{trusts} a block $b$ if there exists a time $t$ and a block $b'$, such that $b'$ arrives to peer $p$ at time time $t$ and is connected by a directed path to $b$ in $G_p(t)$.
  \label{def:trust}
\end{definition}

Note that the notion of trust in Definition \ref{def:trust} is one of distributed agreement.
This is motivated from the fact that building on an existing block requires that a peer has verified that block's content.

\begin{proposition}
  If $b$ is a confirmed block, then all peers in $H$ trust the block $b$.
  \label{prop:trust}
\end{proposition}


The proof is given in Appendix \ref{prf:trust}. 

The Bitcoin whitepaper suggests that as the number of blocks with a directed path to any particular block $b$ increases, a peer can be increasingly confident that all other peers are aware of and have added blocks that reference block $b$.
Definitions \ref{def:confirmed} and \ref{def:trust} and Proposition \ref{prop:trust} aim to capture this notion by looking at the asymptotic structure of the blockchain DAG.

\subsection{Confirmed Blocks and One-Endedness}
\label{sec:one-ended-infinite_confirm}

Observe that a natural requirement on the performance of blockchain systems is that there are infinitely many confirmed blocks. 
Otherwise, the peers consume infinite bandwidth in the limit as $t \to \infty$, yet they only confirm finitely many blocks.


For the rest of this section we assume that the limiting DAG $G(\infty)$ is locally finite. 
Recall that a DAG is locally finite if each vertex has finite degree.
This is done without loss of generality, as we show in Section \ref{sec:one-ended-policies} that the two most natural polices of interest, the tree and throughput-optimal policies, lead to locally finite limiting DAGs.

\begin{lemma}
  Denote by $\widehat{G}(\infty)$ the subgraph of $G(\infty)$ consisting of all confirmed blocks.
  If $G(\infty)$ is locally finite and there are infinitely many confirmed blocks, then $\widehat{G}(\infty)$ is one-ended.
  \label{lem:infinitely-confirmed}
\end{lemma}


The proof is given in Appendix \ref{prf:infinitely-confirmed}. 

Lemma \ref{lem:infinitely-confirmed} shows that if there are infinitely many confirmed blocks, then $G(\infty)$ has at least one end. Moreover, the (infinite) subgraph of all confirmed blocks in $G(\infty)$ is one ended.
In the following lemma we establish that the one-endedness of $G(\infty)$ is a sufficient condition for the existence of infinitely many confirmed blocks.

\begin{lemma}
  If $G(\infty)$ is one-ended and locally finite, then the set of confirmed blocks is infinite.
  \label{lem:one-ended-confirmed}
\end{lemma}

  

The proof is given in Appendix \ref{prf:one-ended-confirmed}. 

The results in Lemmas \ref{lem:infinitely-confirmed} and \ref{lem:one-ended-confirmed} are asymptotic guarantees in the limit as $t \to \infty$; like stability and confirmation, these results cannot be determined at any finite time $t$.
Lemmas \ref{lem:infinitely-confirmed} and \ref{lem:one-ended-confirmed} show that having a one-ended limiting DAG $G(\infty)$ is a desirable property in a blockchain system, as it ensures that the number of confirmed blocks (evaluated in the limit as $t\to\infty$) is infinite.
In Section \ref{sec:one-ended-policies}, we show that under the assumption of stability, both the tree and throughput-optimal policies produce one-ended limiting DAGs.

\section{One-Endedness Under the Tree and Throughput-Optimal Policies}
\label{sec:one-ended-policies}

In this section, we show that stable blockchain systems using the tree and throughput-optimal policies produce one-ended limiting blockchain DAGs.
Recall from Definitions \ref{def:consistent} and \ref{def:stability} that a blockchain system is stable if there exists an infinite sequence of times of consistency $(C_k)_{k \in \mathbb{N}}$ such that $G(C_j) \neq G(C_k)$ if $j \neq k$.

We begin by showing that stable blockchains constructed using the tree and throughput-optimal policies have locally finite limiting DAGs.
\begin{lemma}
  Consider a stable blockchain system.
  If its DAG construction is as per the tree or throughput-optimal policy, its limiting DAG $G(\infty)$ is locally finite.
  \label{lem:locally-finite}
\end{lemma}
  The proof is given in Appendix \ref{prf:locally-finite}.
  

\subsection{Tree Policy}
\begin{theorem}
  Suppose peers add blocks according to the tree policy. 
  If a blockchain system is stable, then its limiting DAG $G(\infty)$ is one-ended.
  \label{thm:tree-one-ended}
\end{theorem}
  The proof is given in Appendix \ref{prf:tree-one-ended}.

\begin{corollary}
  Under the tree policy, if the exogenous arrival and gossip processes together are taken as a Markov process, then $G(\infty)$ is one-ended if the Markov process is positive recurrent.
  \label{cor:markov_dynamics}
\end{corollary}
  The proof is given in Appendix \ref{prf:markov-dynamics}.


\begin{corollary}
  Suppose all peers in a stable blockchain system use the tree policy.
  A block $b$ is confirmed in $G(\infty)$ iff there exists a time of consistency $C$ such that $b$ is on the distinguished path in $G(C)$.
  \label{cor:distinguished_path}
\end{corollary}

  
  

The proof is given in Appendix \ref{prf:distinguished-path}. 
  
Corollary \ref{cor:distinguished_path} shows that for stable blockchain DAGs constructed using the Tree policy, a peer only needs a finite amount of time determine whether or not a particular block will be eventually confirmed. 
If there exists a time of consistency such that $b$ is on the distinguished path in $G(C)$, and additionally the system parameters imply stability (we give conditions for this in Theorem \ref{thm:stability}), then it follows that block $b$, in the limit as $t \rightarrow \infty$, will be confirmed.
Moreover, in stable blockchain systems using the tree policy, Corollary \ref{cor:distinguished_path} provides a necessary and sufficient condition for the confirmation of a block $b$.
In particular, it shows that confirmation is equivalent to the existence of an infinite sequence of blocks $(b_k)_{k \in \mathbb{N}}$ such that there is a path in $G(\infty)$ from $b_k$ to $b$ for all $k \in \mathbb{N}$.
In general, this condition is only a necessary condition, but not sufficient.
Examples of blockchains using the tree policy are Bitcoin and Ethereum, the two most commonly used blockchain implementations \cite{nakamoto2008bitcoin, buterin2013ethereum}.
In Sections \ref{sec:interpretations} and \ref{sec:simulations} we use the fact that for stable blockchain systems using the tree policy, confirmation can be determined in finite time in order to identify and numerically estimate several network parameters related to stability of the Bitcoin P2P network.

\subsection{Throughput-Optimal Policy}
\begin{theorem}
  Let peers add blocks according to the throughput-optimal policy.
  If the system is stable, then the limiting DAG $G(\infty)$ is one-ended.
  \label{thm:throughput-optimal-one-ended}
\end{theorem}
  The proof is given in Appendix \ref{prf:throughput-optimal-one-ended}.

This theorem does not follow from Theorem \ref{thm:tree-one-ended}, since in general, it is not true that adding or removing countably many edges to or from a one-ended DAG results again in a one-ended DAG.
A counterexample in each direction is given in Figure \ref{fig:ends}.

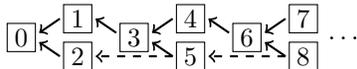
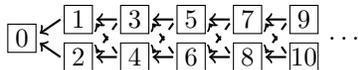
\begin{figure}[htp]
    \centering
    \captionsetup{justification=centering}
    \begin{subfigure}[!bt]{0.8\columnwidth}
    {\centering
        \begin{tikzpicture}[scale=0.5]
        \draw [fill=white] (0, 0) rectangle (0.75, 0.75);        \node at (0.375, 0.375) {0};
        \draw [fill=white] (1.5, 0.5) rectangle (2.25, 1.25);        \node at (1.875, 0.875) {1};
        \draw [fill=white] (1.5, 0.25) rectangle (2.25, -0.5);        \node at (1.875, -0.125) {2};
        \draw [fill=white] (3, 0) rectangle (3.75, 0.75);        \node at (3.375, 0.375) {3};
        \draw [fill=white] (4.5, 0.5)
        rectangle (5.25, 1.25);        \node at (4.875, 0.875) {4};
        \draw [fill=white] (4.5, 0.25)
        rectangle (5.25, -0.5);        \node at (4.875, -0.125) {5};
        \draw [fill=white] (6, 0) rectangle (6.75, 0.75);        \node at (6.375, 0.375) {6};
        \draw [fill=white] (7.5, 0.5) rectangle (8.25, 1.25);        \node at (7.875, 0.875) {7};
        \draw [fill=white] (7.5, 0.25) rectangle (8.25, -0.5);        \node at (7.875, -0.125) {8};
        
        \draw[->, thick] (1.4, 0.875) -- (0.85, 0.5);
        \draw[->, thick] (1.4, -0.125) -- (0.85, 0.25);
        \draw[->, thick] (2.9, 0.5) -- (2.35, 0.875);
        \draw[->, thick] (4.4, 0.875) -- (3.85, 0.5);
        \draw[->, thick] (4.4, -0.125) -- (3.85, 0.25);
        \draw[->, thick] (5.9, 0.5) -- (5.35, 0.875);
        \draw[->, thick] (7.4, 0.875) -- (6.85, 0.5);
        \draw[->, thick] (7.4, -0.125) -- (6.85, 0.25);
        
        \draw[->, dashed, thick] (4.4, -0.125) -- (2.35, -0.125);
        \draw[->, dashed, thick] (7.4, -0.125) -- (5.35, -0.125);
        
        \node at (9.0, 0.375) {\ldots};
        \end{tikzpicture}
        \par}
    \caption{Adding the dashed edges increases the number of ends from 1 to 2.}
\end{subfigure}

\begin{subfigure}[!bt]{0.8\columnwidth}
    {\centering
        \begin{tikzpicture}[scale=0.5]
        \draw [fill=white] (0, 0) rectangle (0.75, 0.75);        \node at (0.375, 0.375) {0};
        \draw [fill=white] (1.5, 0.5) rectangle (2.25, 1.25);        \node at (1.875, 0.875) {1};
        \draw [fill=white] (1.5, 0.25) rectangle (2.25, -0.5);        \node at (1.875, -0.125) {2};
        \draw [fill=white] (3, 0.5) rectangle (3.75, 1.25);        \node at (3.375, 0.875) {3};
        \draw [fill=white] (3, 0.25) rectangle (3.75, -0.5);        \node at (3.375, -0.125) {4};
        \draw [fill=white] (4.5, 0.5)
        rectangle (5.25, 1.25);        \node at (4.875, 0.875) {5};
        \draw [fill=white] (4.5, 0.25)
        rectangle (5.25, -0.5);        \node at (4.875, -0.125) {6};
        \draw [fill=white] (6, 0.5) rectangle (6.75, 1.25);        \node at (6.375, 0.875) {7};
        \draw [fill=white] (6, 0.25) rectangle (6.75, -0.5);        \node at (6.375, -0.125) {8};
        \draw [fill=white] (7.5, 0.5) rectangle (8.25, 1.25);        \node at (7.875, 0.875) {9};
        \draw [fill=white] (7.5, 0.25) rectangle (8.25, -0.5);        \node at (7.875, -0.125) {10};
        
        \draw[->, thick] (1.4, 0.875) -- (0.85, 0.5);
        \draw[->, thick] (1.4, -0.125) -- (0.85, 0.25);
        \draw[->, thick] (2.9, 0.875) -- (2.35, 0.875);
        \draw[->, thick] (2.9, -0.125) -- (2.35, -0.125);
        \draw[->, thick] (4.4, 0.875) -- (3.85, 0.875);
        \draw[->, thick] (4.4, -0.125) -- (3.85, -0.125);
        \draw[->, thick] (5.9, 0.875) -- (5.35, 0.875);
        \draw[->, thick] (5.9, -0.125) -- (5.35, -0.125);
        \draw[->, thick] (7.4, 0.875) -- (6.85, 0.875);
        \draw[->, thick] (7.4, -0.125) -- (6.85, -0.125);
        
        \draw[->, dashed, thick] (2.9, -0.025) -- (2.35, 0.775);
        \draw[->, dashed, thick] (2.9, 0.775) -- (2.35, -0.025);
        \draw[->, dashed, thick] (4.4, -0.025) -- (3.85, 0.775);
        \draw[->, dashed, thick] (4.4, 0.775) -- (3.85, -0.025);
        \draw[->, dashed, thick] (5.9, -0.025) -- (5.35, 0.775);
        \draw[->, dashed, thick] (5.9, 0.775) -- (5.35, -0.025);
        \draw[->, dashed, thick] (7.4, -0.025) -- (6.85, 0.775);
        \draw[->, dashed, thick] (7.4, 0.775) -- (6.85, -0.025);
        
        \node at (9.0, 0.375) {\ldots};
        \end{tikzpicture}
        \par}
    \caption{Removing the dashed edges increases the number of ends from 1 to 2.}
\end{subfigure}
    \captionsetup{justification=justified, singlelinecheck=false}
    \caption{Examples of one-ended DAGs where adding and removing countably many edges increases the number of ends from 1 to 2.}
    \label{fig:ends}
\end{figure}

\begin{corollary}
  Under the throughput-optimal policy, if the exogenous arrival and gossip processes together are taken as a Markov process, then $G(\infty)$ is one-ended if the Markov process is positive recurrent.
\label{cor:throughput-markov}
\end{corollary}
The proof is given in Appendix \ref{prf:throughput-markov}.

\begin{corollary}
  In a stable blockchain system using the throughput-optimal policy, all blocks in $G(\infty)$ are confirmed blocks.
  \label{cor:throughput-all-confirmed}
\end{corollary}
The proof is given in Appendix \ref{prf:throughput-all-confirmed}.

As a result of Corollary \ref{cor:throughput-all-confirmed}, we note that all blocks will eventually be confirmed.
Hence, we denote this policy as \emph{throughput-optimal}.

Lemma \ref{lem:locally-finite} and Theorems \ref{thm:tree-one-ended} and \ref{thm:throughput-optimal-one-ended} establish that stable \\ blockchain systems using the tree and throughput-optimal policies have one-ended limiting DAGs. 
In particular, Corollary \ref{cor:distinguished_path} shows that under the tree policy, there exists an infinite path consisting of all (infinitely many) confirmed blocks; Corollary \ref{cor:throughput-all-confirmed} shows that under the throughput-optimal policy, stability implies that all blocks will eventually be confirmed.


\section{Stability Analysis}
\label{sec:stability}
In this section, we provide quantitative bounds on the block arrival rate as a function of the structure of the peer-to-peer network $H$ so that the block communication process is stable.
Recall that stability is a property of the communication infrastructure supporting the blockchain system.
It is necessary to have a stable communication process among peers so that the blockchain system can confirm infinitely many blocks, as is shown in Lemmas \ref{lem:one-ended-confirmed} and \ref{lem:locally-finite} and Theorems \ref{thm:tree-one-ended} and \ref{thm:throughput-optimal-one-ended}.

For the rest of this section we assume that $H$ is an arbitrary, fixed, undirected, and connected network on $N$ peers.
We recall that $A$ is the arrival process of blocks into our system, and for any $m \in \mathbb{Z}$, we denote by $A_m$ the time of the $m$-th arrival to the system. 
In particular, we let $X_{[m, n]}(A)$ denote the earliest time when $B_p(t) = \{m, \ldots, n\} \ \forall p$, when the arrival process is restricted only to the arrivals $A_m, \ldots, A_n$. 
In other words when considering $X_{[m, n]}(A)$, arrivals begin at time $A_m$, and no arrival occurs after time $A_n$.
Technically, $A$ is a marked point process on $\mathbb{R}_+$, with the convention that the first arrival after time 0 corresponds to $A_1$.
Thus, for any $m \in \mathbb{Z}$, the $m$-th arrival occurs at time $A_m$ and the mark of the $m$-th point (occurring at time $A_m$) consists of the following:
\begin{enumerate}
    \item The peer $p_m \in \{1, \ldots, N\}$ at which $A_m$ arrives in the system.
    \item The sequence of points of the processes $(T_p - A_m)_{p \in \{1, \ldots, N\}}$. 
    In words, this is the set of all potential communication times between peers, shifted by $A_m$.
    A pictorial representation of this part of the marks of the process $A$ is in Figure \ref{fig:marks-realization}.
\end{enumerate}
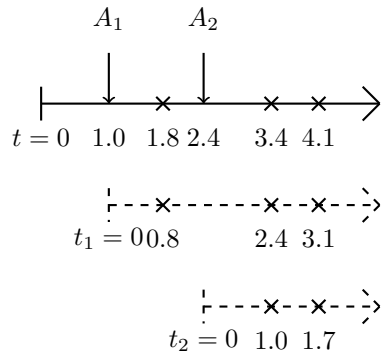
\begin{figure}[ht!]
  {\centering
    \begin{tikzpicture}[scale=0.9]
    \draw[thick] (0, 0.25) -- (0, -0.25);
    \node at (0, -0.5) {$t=0$};
    \draw[thick] (0, 0) -- (5, 0);
    \draw[thick] (4.75, 0.25) -- (5, 0);
    \draw[thick] (4.75, -0.25) -- (5, 0);
    
    \draw[thick, ->] (1, 0.75) -- (1, 0);
    \node at (1, -0.5) {$1.0$};
    \node at (1, 1.25) {$A_1$};
    
    \draw[thick] (1.7, 0.1) -- (1.9, -0.1);
    \draw[thick] (1.7, -0.1) -- (1.9, 0.1);
    \node at (1.8, -0.5) {$1.8$};
    
    \draw[thick, ->] (2.4, 0.75) -- (2.4, 0);
    \node at (2.4, -0.5) {$2.4$};
    \node at (2.4, 1.25) {$A_2$};
    
    \draw[thick] (3.3, 0.1) -- (3.5, -0.1);
    \draw[thick] (3.3, -0.1) -- (3.5, 0.1);
    \node at (3.4, -0.5) {$3.4$};
    
    \draw[thick] (4.0, -0.1) -- (4.2, 0.1);
    \draw[thick] (4.0, 0.1) -- (4.2, -0.1);
    \node at (4.1, -0.5) {$4.1$};


    \draw[thick, dashed] (1.0, -1.25) -- (1.0, -1.75);
    \node at (1.0, -2) {$t_1 = 0$};
    \draw[thick, dashed] (1.0, -1.5) -- (5, -1.5);
    \draw[thick, dashed] (4.75, -1.75) -- (5, -1.5);
    \draw[thick, dashed] (4.75, -1.25) -- (5, -1.5);
    
    \draw[thick] (1.7, -1.4) -- (1.9, -1.6);
    \draw[thick] (1.7, -1.6) -- (1.9, -1.4);
    \node at (1.8, -2) {$0.8$};
    
    \draw[thick] (3.3, -1.4) -- (3.5, -1.6);
    \draw[thick] (3.3, -1.6) -- (3.5, -1.4);
    \node at (3.4, -2) {$2.4$};
    
    \draw[thick] (4.0, -1.4) -- (4.2, -1.6);
    \draw[thick] (4.0, -1.6) -- (4.2, -1.4);
    \node at (4.1, -2) {$3.1$};
    

    \draw[thick, dashed] (2.4, -2.75) -- (2.4, -3.25);
    \node at (2.4, -3.5) {$t_2 = 0$};
    \draw[thick, dashed] (2.4, -3) -- (5, -3);
    \draw[thick, dashed] (4.75, -3.25) -- (5, -3);
    \draw[thick, dashed] (4.75, -2.75) -- (5, -3);
    
    \draw[thick] (3.3, -2.9) -- (3.5, -3.1);
    \draw[thick] (3.3, -3.1) -- (3.5, -2.9);
    \node at (3.4, -3.5) {$1.0$};
    
    \draw[thick] (4.0, -2.9) -- (4.2, -3.1);
    \draw[thick] (4.0, -3.1) -- (4.2, -2.9);
    \node at (4.1, -3.5) {$1.7$};

    \end{tikzpicture}
    \par}
    \caption{An example realization of the temporal marks for the for arrivals $A_1$ and $A_2$.
    Arrivals are represented by vertical arrows and the points of the transmission processes $T_p, p \in [N]$ are represented by $\times$ signs.}
    \label{fig:marks-realization}
\end{figure}

We add the following standard assumptions made in the analysis of P2P networks (~\cite{ganeshnotes, sanghavi2007gossiping,  shah2009gossip}), and state our main result regarding stability stated below in Theorem \ref{thm:stability}.

\begin{assumption}
  The arrival process $A$ is an arbitrary stationary point process with intensity $\lambda$ such that blocks arrive uniformly at each peer.
  \label{assump:stationary_arrivals}
\end{assumption}
\begin{assumption}
  The outgoing communications $T_p$ from each peer $p \in \{1, \ldots, N\}$, occur as a rate 1 Poisson point process; the receiving peer for each communication is chosen uniformly and independently at random from the neighbors of $p$ in $H$.
  The communication process $T_p$ is independent from the communications of all other peers $T_q$ as well as from the arrival point process $A$.
  \label{assump:poisson_gossip}
\end{assumption}

\begin{definition}
    Let $H$ be an undirected network on a vertex set $V$ and let $S \subseteq V$ be a subset of vertices. 
    The conductance $\phi_{H}^{(S)}$ of the set $S$ is given by 
    $$
    \phi_{H}^{(S)} = \frac{\sum_{p \in S, q \in S^C}\frac{1}{d(p)}\mathbf{1}_{pq}}{\frac{1}{N}|S||S^C|},
    $$
    where $\mathbf{1}_{pq}$ is an indicator random variable for the edge $(p, q)$ being in the edge set of $H$, $d(p)$ is the degree of $p$, and $|\cdot|$ is the cardinality of the set $\cdot$.
    The conductance of the network $H$, denoted by $\phi_H$, is then defined as 
    \begin{align*}
        \phi_H = \inf_{S\subseteq V} \phi_H^{(S)}.
    \end{align*}
    \label{def:conductance}
\end{definition}

The conductance of a network $H$ is an indicator of how much information propagation on $H$ is affected by bottlenecks -- networks with higher conductance are less affected by bottlenecks.
Observe that the numerator of the expression for $\phi_H^{(S)}$ represents the rate of attempted communication on edges from $S$ to $S^C$.
Also note that by choosing $S$ to be a cut consisting of a single vertex, we can obtain the bound $\phi_H \leq \frac{N}{N-1}$; equality is achieved if $H$ is complete.

\begin{theorem}
   Let $\phi_H^{(S)}$ be the conductance of the cut $S$ and $\phi_H$ be the  conductance of the Peer-to-Peer network $H$ on the vertex set $\{1,\cdots,N\}$ and
    the point processes $(T_p)_{p=1}^N$ are mutually independent for all $p \in \{1, 2, \ldots, N\}$. Then there exists $\mu$ satisfying $$\frac{\phi_H}{2\log N} \leq \mu \leq \inf_{S \subset H}|S^C|\phi_H^{(S)},$$ such that the blockchain system is stable for all arrival rates $\lambda$ satisfying $0 < \lambda < \mu < \infty$.
   \label{thm:stability}
\end{theorem}
The proof is given in Appendix \ref{prf:stability}.

This theorem provides quantitative bounds on the maximum stable block arrival rate in terms of the conductance of the peer-to-peer communication network $H$. 

\begin{remark}
  In Section \ref{sec:model} and Theorem \ref{thm:stability}, we assume that the communication processes $T_p$ are rate $1$.
  More generally, if the processes $T_p$ are rate $B$, the bounds in Theorem \ref{thm:stability} become
  $$
    \frac{B\phi_H}{2\log N} \leq \mu \leq \inf_{S \subset H}B|S^C|\phi_H^{(S)}.
  $$
  This corresponds to giving each peer a communication bandwidth of $B$ blocks per second instead of unit bandwidth.
  \label{remark-bandwidth}
\end{remark}

In order to prove Theorem \ref{thm:stability}, we use the \emph{monotone separability} framework \cite{baccelli1995saturation}, which is made precise in Lemma \ref{lem:monotone_separable}. 
An exposition of the monotone separability framework is provided in the survey paper \cite{foss2004overview}.

\begin{lemma}
  For all $n \geq m,$ $X_{[m, n]}$ is monotone separable; namely, it satisfies the following.
  \begin{enumerate}
      \item For all $n \geq m,$ $X_{[m, n]}$ is causal; namely $$X_{[m, n]}(A) \geq A_n.$$
      \item For all $n \geq m,$ $X_{[m, n]}$ is externally monotonic; namely $$X_{[m, n]}(A') \geq X_{[m, n]}(A)$$ if $A'$ is a point process such that $A'_m \geq A_m$ for all $m \in \mathbb{N}$.
      \item For all $n \geq m,$ $X_{[m, n]}$ is homogeneous; namely $$X_{[m, n]}(A + c) = X_{[m, n]}(A) + c \ \forall c \in \mathbb{R}.$$
      \item For all $n \geq m,$ $X_{[m, n]}$ is separable; namely $$X_{[m, n]}(A) = X_{[l+1, n]}(A)$$ if $X_{[m, l]} \leq A_{l+1}$.
  \end{enumerate}
  \label{lem:monotone_separable}
\end{lemma}
The proof is given in Appendix \ref{prf:monotone_separable}.

Define $X_n$ to be the earliest time $t$ when $B_p(t) = \{1, \ldots, n\}$ for all peers $p \in \{1, \ldots, N\}$, such that all of the arrivals $A_1, \ldots, A_n$ arrive at time $t = 0$ and no arrival occurs after $A_n$.
$X_n$ is called the maximal dater in the queueing theory literature.
We state below a result of Baccelli and Foss \cite{baccelli1995saturation} regarding the stability of monotone separable systems.
\begin{theorem}[Baccelli and Foss \cite{baccelli1995saturation}]
  For a monotone separable system, the limit 
  $$0 \leq \mu^{-1} := \lim_{n \to \infty}\frac{X_n}{n} = \lim_{n\to\infty}\frac{\mathbb{E}[X_n]}{n}$$
  exists almost surely. Moreover, the system is stable if the arrival rate satisfies $\lambda < \mu$ and unstable if $\lambda > \mu$.
  \label{thm:monotone-separable}
\end{theorem}

The key insight in the proof of Theorem \ref{thm:stability} is to use Theorem \ref{thm:monotone-separable} to bound the constant $\mu^{-1}$ as follows.
First, we shift all block arrivals such that each block arrives at the instant the previous block is known to all peers; this provides an upper bound on $\mu^{-1}$.
Next, we find a lower bound on $\mu^{-1}$ by shifting blocks arrivals such that all blocks are present in the system at time $t = 0$ and lower bounding the time, for any set $S$, for all blocks initially contained in $S$ to be known to some peer in $S^C$.

Theorem \ref{thm:stability} shows that a guaranteed stability condition is $\lambda < \frac{\phi_H}{2\log N}$ and that the true stability region for a blockchain system is upper bounded by the condition $\lambda < \inf_{S \subset H}\phi_H^{(S)}$.
In particular, we find that the critical rate $\mu$ depends on both $N$ and the network topology of $H$.
As shown previously in Lemmas \ref{lem:one-ended-confirmed} and \ref{lem:locally-finite} and in Theorems \ref{thm:tree-one-ended} and \ref{thm:throughput-optimal-one-ended}, the stability of the communication dynamics guarantees the one-endedness of the blockchain DAG $G(\infty)$, which in turn implies the existene of infinitely many confirmed blocks in the limit as $t\to\infty$.

\subsection{Scalability}
In this section, we provide an illustrative example using Theorem \ref{thm:stability} to reject the scalability of a sequence of peer-to-peer networks (note that Theorem \ref{thm:stability} cannot be used to ensure scalability).
In addition, we show that no network on $N$ peers can support an arrival rate greater than $\frac{N}{N-1}$; hence even if a sequence of networks is scalable, it is impossible to support a block arrival rate of greater than $2$.

\subsubsection{Large Stars are Not Scalable}
From Theorem \ref{thm:stability}, we note that a necessary condition for scalability of the sequence of networks $(H_k)_{k \in \mathbb{N}}$ is that the bound $\inf_{S \subset H_k} |S^C|\phi_{H_k} \not\to 0$ as $k \to \infty$.

Recall that a \emph{star} on $k+1$ vertices is a connected graph consisting of a single vertex of degree $k$, and $k$ vertices each of degree $1$ (see Figure \ref{fig:star} for a pictorial example).
We show that if $H$ is a star with $k+1$ peers, the stability region is bounded above by $\frac{k+1}{k^2}$, which decreases to $0$ with $k$.
Thus, sequences of networks containing large stars are not scalable as per Definition \ref{def:scalability}.
The non-scalability of stars can be intuitively explained as follows.
Suppose that there are $k$ leaves in a star.
As the central node evenly shares unit bandwidth among the $k$ leaves, each leaf receives incoming blocks (from other peers) at a decreasing rate as $k \to \infty$.
In order to have stability, the total block arrival rate must be small enough that all blocks can be transmitted to the leaf with the reduced bandwidth.
Thus, the stability region must decrease to $0$, since $\frac1k$ decreases to $0$.
Indeed, suppose $H$ is a star consisting of $N$ peers and consider a cut $S$ such that $S^C$ consists of a single leaf $p$.
Then the stability of $H$ is bounded above by:
$$
\mu \leq \inf_{S \subset H}|S^C|\phi_H^{(S)} \leq (1)\frac{\sum_{q \in S}\frac{1}{d(q)}\mathbf{1}_{pq}}{\frac{1}{N}|S||S^C|} = (1)\frac{\frac{1}{N-1}}{\frac{1}{N}|S||S^C|} = (1)\frac{\frac{1}{N-1}}{\frac{1}{N}(N-1)(1)} = \frac{N}{(N-1)^2} \approx \frac{1}{N-1}.
$$
Thus, we can approximate the upper bound on the the stability region of a star by the incoming bandwidth to any leaf, which agrees with the intuitive explanation of non-scalability.

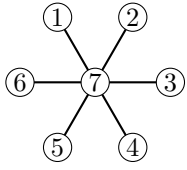
\begin{figure}
    \centering
    {\centering
  \begin{tikzpicture}[scale=0.5]
  
  \draw [-, thick] (-1, 0) -- (0, -1.732);
  \draw [-, thick] (1, 0) -- (0, -1.732);
  \draw [-, thick] (2, -1.732) -- (0, -1.732);
  \draw [-, thick] (1, -3.464) -- (0, -1.732);
  \draw [-, thick] (-1, -3.464) -- (0, -1.732);
  \draw [-, thick] (-2, -1.732) -- (0, -1.732);
  
  \draw[fill=white] (-1, 0) circle [radius=0.375];
  \node at (-1, 0) {$1$};
  \draw[fill=white] (1, 0) circle [radius=0.375];
  \node at (1, 0) {$2$};
  \draw[fill=white] (2, -1.732) circle [radius=0.375];
  \node at (2, -1.732) {$3$};
  \draw[fill=white] (1, -3.464) circle [radius=0.375];
  \node at (1, -3.464) {$4$};
  \draw[fill=white] (-1, -3.464) circle [radius=0.375];
  \node at (-1, -3.464) {$5$};
  \draw[fill=white] (-2, -1.732) circle [radius=0.375];
  \node at (-2, -1.732) {$6$};
  \draw[fill=white] (0, -1.732) circle [radius=0.375];
  \node at (0, -1.732) {$7$};
  \end{tikzpicture}  
}
    \caption{A star on $7$ vertices/peers.}
    \label{fig:torus}
\end{figure}

\subsubsection{Constant Per-Peer Arrival Rates Cannot be Supported}
Recall that the global arrival process is thinned to each peer such that the rates of blocks arrivals to all peers are equal.
We show that for a network of $N$ peers, the per-peer block arrival rate cannot exceed $\frac{1}{N-1}$.

\begin{proposition}
  Let $H$ be a peer-to-peer network between $N$ peers.
  The critical arrival rate $\mu$ of new blocks does not exceed $\frac{N}{N-1}$.
  In particular, in a stable blockchain system, no peer adds blocks at a rate greater than $\frac{1}{N-1}$.
  \label{prop:per-peer-arrival-rate}
\end{proposition}
\begin{proof}
  Let $p$ be a a peer in $H$ with minimal degree (the existence of such a peer is due to the well-ordering principle on $\mathbb{N}$).
  Consider a cut $S$ so that $S^C = \{p\}$.
  As $p$ is a vertex of minimal degree, for all peers $q$ such that $p$ and $q$ are connected by an edge in $H$, the rate of attempted communication from $p$ to $q$ is at least as great as the rate of attempted communication from $q$ to $p$.
  Thus, we can bound the stability region of the network $H$:
  $$
    \mu \leq \inf_{S \subset H}|S^C|\phi_H^{(S)} \leq (1)\frac{\sum_{q \in S}\frac{1}{d(q)}\mathbf{1}_{pq}}{\frac{1}{N}|S||S^C|} \leq (1)\frac{\sum_{q \in S}\frac{1}{d(p)}\mathbf{1}_{pq}}{\frac{1}{N}|S||S^C|} = (1)\frac{1}{\frac1N(N-1)(1)} = \frac{N}{N-1}.
  $$
\end{proof}

Notice that Propositon \ref{prop:per-peer-arrival-rate} makes no assumptions on the network $H$, other than the number of vertices.
Thus, even in a scalable sequence of peer-to-peer networks, the per-peer arrival rate decreases to zero. In other words, for \emph{any} growing sequence of networks, no positive per-peer arrival rate of blocks to the network nodes can be supported on all networks in the sequence.
\section{Blockchain System Design Insights}
\label{sec:interpretations}
In this section we discuss design insights for blockchain systems based on the results of Section \ref{sec:stability}.
We assume that the arrival process $A$ and communication processes $T_p, p \in \{1, \ldots, N\}$ are stationary and ergodic on a P2P network $H$ so that all metrics are time-invariant.

\subsection{One-Endedness as a Form of Distributed Consensus}
\label{sec:one_end_dist_consensus}
The goal of the blockchain paradigm is to enable a set of anonymous peers to agree on a distributed ledger, where information is stored in discrete units called blocks.
In particular, as an infinite amount of bandwidth is consumed in the limit as $t\to\infty$, a natural requirement is that the number of confirmed blocks also tends to infinity as $t\to\infty$.

We show in Lemma \ref{lem:one-ended-confirmed} that the one-endedness and local finiteness of the limiting DAG $G(\infty)$ is a sufficient condition for the existence of infinitely many confirmed blocks.
We then show, in Theorems \ref{thm:tree-one-ended} and \ref{thm:throughput-optimal-one-ended}, that under the assumption of stability,  blockchains using the tree and throughput-optimal policies have one-ended limiting DAGs.
For the throughput-optimal policy, we find that all blocks are eventually confirmed.
For the tree policy, we find that an external observer who is aware of the states of $G_p(t)$ for all peers $p \in \{1, \ldots, N\}$ can determine, in finite time, whether or not any particular block will be eventually confirmed.
In particular, such determinations occur at times of consistency.
This ties the confirmation of blocks under the tree policy to the network dynamics on the peer-to-peer network $H$.

We thus identify the following metrics on the network $H$.
As the underlying peer-to-peer networks in the Bitcoin and Ethereum blockchain implementations (\cite{nakamoto2008bitcoin, buterin2013ethereum}) use the tree policy, these metrics provide insight into the behavior of those systems.
In Section \ref{sec:simulations}, we use these metrics in a simulation environment to numerically estimate key network properties of the Bitcoin peer-to-peer network.

\subsection{Performance Metrics}

In this section, we identify some quantitative system performance metrics which further characterize the performance of the blockchain system when it is stable.
Recall Theorem \ref{thm:stability} provides bounds on the maximum block arrival rate $\mu$ to guarantee stability. 



\noindent{\textbf{Time to Consistency}} -- 
This is defined as the minimum time a peer should wait after a block arrival $b$ (in expectation under steady state) such that all other peers have knowledge of $b$.
Our simulation results (in Figure \ref{fig:TtC-BlockRate}) suggest that the time to consistency increases monotonically with the block arrival rate as expected. 
In practice, peers wait for six future blocks to arrive after any given block before trusting that is a part of the ledger \cite{confirmation}. 
This choice is made ad-hoc, assuming a fixed block arrival rate of roughly one in every $10$ minutes \cite{nakamoto2008bitcoin}.
Our simulation results in Figure \ref{fig:TtC-BlockRate} provide a quantitative way of choosing the threshold as a function of the block arrival rate.


\noindent{\textbf{Cycle Length}} -- 
Cycle length is defined as the sum of the mean (under steady state) busy period and mean idle period. 
A busy period is the time it takes for an inconsistent system to reach consistency -- that is, the mean length of a busy period is equal to the time to consistency metric. An idle period is the length of time for which a consistent system remains consistent.

Observe that as the cycle length is at least the mean idle time, it goes to infinity as the block arrival rate goes to zero.
Thus, the cycle length metric captures the trade-off between the time to consistency (which goes to $0$ as the block arrival rate goes to $0$) and the block arrival rate.

Our simulation results indicate that the cycle length may be a convex function of block arrival rates (Figure \ref{fig:CycleLength-Blockrate}); thus on a given P2P network $H$ there may be a unique optimal block arrival rate that minimizes the cycle length.
Of note, the results in Figure \ref{fig:CycleLength-Blockrate} identify that the cycle length may satisfy two key robustness properties.
The first is that for a fixed $N$, there is a wide range of block arrival rates for which the cycle length is approximately constant.
The second is that there exists a range of block arrival rates for which the cycle length is nearly invariant to the number of peers in the network.


\noindent{\textbf{Consistency Fraction}} -- 
This metric is defined as the expected fraction of peers $p \in \{1, \ldots, N\}$ for which $B_p(t) = B(t)$ at any time $t$ (in steady state) -- we call these peers \emph{consistent}. Notice that this metric does not depend on time as the system is assumed to be in steady state.
In particular, this metric provides a lower bound on the growth rate of the distinguished path of tree policy blockchains -- all consistent peers add blocks to the distinguished path, but some inconsistent peers may do so as well.
In an implementation such as the Bitcoin blockchain, arrivals at inconsistent peers contribute to wasted mining power and energy consumption.

\noindent{\textbf{Growth Rate of the Distinguished Path}} -- For tree policy blockchains, the growth rate of the distinguished path characterizes the \emph{progress} of the system in steady state; namely the rate at which eventually confirmed blocks are added.
In particular, for tree policy blockchains, the growth rate of the distinguished path characterizes the block throughput, as only blocks on the distinguished path will eventually be confirmed.

Our simulation results (Figure \ref{fig:GrowthRate-BlockRate}) indicate that despite a monotonically decreasing consistency fraction, the growth rate of the distinguished path increases to a maximum before decaying.
This suggests that there may be a unique block arrival rate which maximizes the growth rate of the distinguished path.
We note that the arrival rate which minimizes the cycle length appears to be less than the arrival rate which maximizes the growth rate of the distinguished path.

When the block arrival rate is small, almost all blocks are on the distinguished path and very few are \emph{orphaned} (not on the distinguished path).
Therefore, a small increase in arrival rate increases the throughput.
On the contrary, for large block arrival rates, nearly all blocks are orphaned; hence increasing the block arrival rate decreases throughput.

\noindent{\textbf{Age of Information}} -- 
This is a relatively new metric in the queueing theory literature and has had significant impact on scheduling algorithms \cite{kaul2012real}.
We measure the age of information for a peer $p$ in discrete units, where an age of 0 indicates that the peer is consistent, an age of 1 indicates that the peer is 1 block away from being consistent, \textit{etc.}.
As expected, the age of information increases monotonically with block arrival rates (see Figure \ref{fig:NumBlocksBehind-BlockRate}).
The age of information is inversely related to the consistency fraction.

\subsection{Trade-Offs Between the Metrics}
We note that there are various trade-offs between the metrics identified above, which are confirmed by the simulations in Section \ref{sec:simulations}.
In particular, we note that the time to consistency, consistency fraction, and age of information are all optimized as the block rate decreases to zero.
However, at the expense of performance with respect to these metrics, increasing the block rate  can decrease the cycle length and increase the growth rate of the distinguished path in tree policy blockchains, as shown in Figures \ref{fig:CycleLength-Blockrate} and \ref{fig:GrowthRate-BlockRate}.

The cycle length and growth rate are more sophisticated metrics and characterize the temporal dynamics of a blockchain system.
The growth rate of the distinguished path also characterizes the generation rate of orphaned blocks.
Minimizing the rate of orphaned blocks is desirable as orphaned blocks may pose security threats in applications such as cryptocurrencies \cite{gervais2016security}.
Nevertheless, a detailed analysis of the security of blockchains is application-specific and therefore out of the scope of this paper \cite{pass2017analysis, garay2015bitcoin}, as our goal is to study aspects of blockchain which are universal to all applications.

For the core blockchain protocol, the block arrival rate is the only system parameter that can be chosen by a system designer.
This underscores the need for further study on the performance of blockchain systems with regard to the metrics identified in this work in order to improve current blockchain systems and design future ones.
\section{Simulation Results}
\label{sec:simulations}
In this section, we numerically analyze the blockchain system under two settings -- a synthetic data setting and a real data setting comprising of block arrival data of the Bitcoin network \cite{blockchair2020data, nakamoto2008bitcoin}. 

\subsection{Synthetic Data}

We numerically study our stochastic network model and further characterize its performance whenever it is stable. 
In particular, we use the metrics identified in Section \ref{sec:interpretations} to gain further insights into the network behavior. 
This complements our theoretical result which gives bounds on the stability region. 
We analyze the metrics identified in Section \ref{sec:interpretations} with respect to varying block arrival rates, using synthetic parameters for the network.
This is only possible with synthetic data, as the real data set consists of a single arrival process of blocks and thereby we cannot use it to directly assess the impact of varying the block arrival rate on system performance.

\noindent {\bf Simulation Setup} -- 
We consider three different P2P networks comprising of the complete network on $10, 20, 30$ peers (nodes). 
In all these cases, each peer attempts a communication at rate $1$.
All simulations were run for $500$ cycles, with $30$ independent simulations for each block arrival rate.
The error bars represent $95\%$ confidence intervals.
Theorem \ref{thm:stability} gives bounds on the stability region in the three cases as $0.47 \leq \mu_{10} \leq 1.1$, $0.35 \leq \mu_{20} \leq 1.05$ and $0.30 \leq \mu_{30} \leq 1.03$ respectively.
Our simulation suggests that the true critical value is closer to the lower bound in all of these cases.


\noindent{\textbf{Time to Consistency}} --
\begin{figure}[ht!]
    \centering
    \includegraphics[width=0.65\columnwidth]{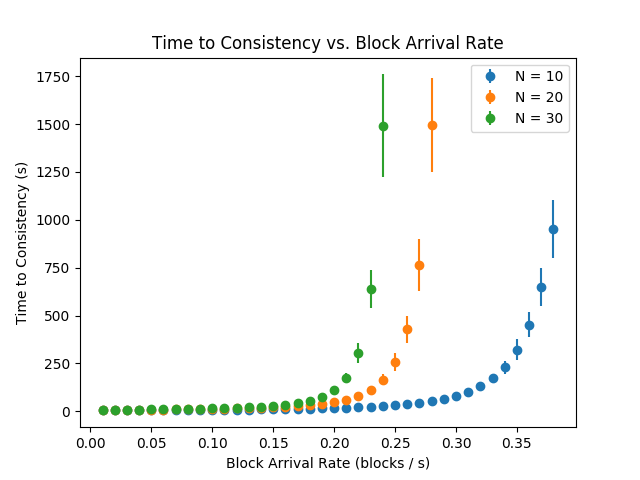}
    \caption{The mean time to consistency.}
    \label{fig:TtC-BlockRate}
\end{figure}
Figure \ref{fig:TtC-BlockRate} shows the effect of increasing the block arrival rate on the time to consistency.
As expected, we observe that the the time to consistency grows monotonically with block arrival rate.

\noindent{\textbf{Cycle Length}} --
\begin{figure}[ht!]
    \centering
    \includegraphics[width=0.65\columnwidth]{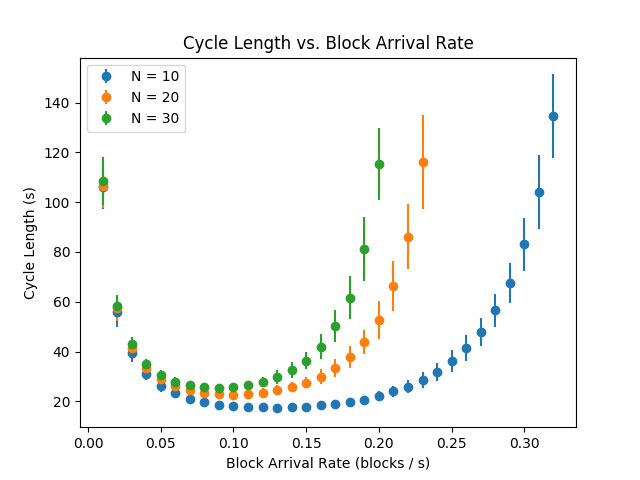}
    \caption{The mean cycle length.}
    \label{fig:CycleLength-Blockrate}
\end{figure}
In Figure \ref{fig:CycleLength-Blockrate} we depict the average cycle length, which is the sum of the mean time to consistency and the mean length of consistency.
Observe that the cycle length has asymptotes to infinity at both $0$ and $\mu$, which are observed in Figure \ref{fig:CycleLength-Blockrate}.
We observe that the cycle length appears to be a convex function of block arrival rates for complete networks and is nearly flat near its infimum.
This suggests that when $H$ is a complete network, there is a wide range of block arrival rates for which the cycle length is approximately constant, indicating a sense of robustness for block arrival rates for the designer of the system to choose.
The figure also suggests there is a reasonably large set of block rates for which the cycle length is fairly robust to changes in the number of peers.

\noindent{\textbf{Consistency Fraction}} --
\begin{figure}[ht!]
    \centering
    \includegraphics[width=0.65\columnwidth]{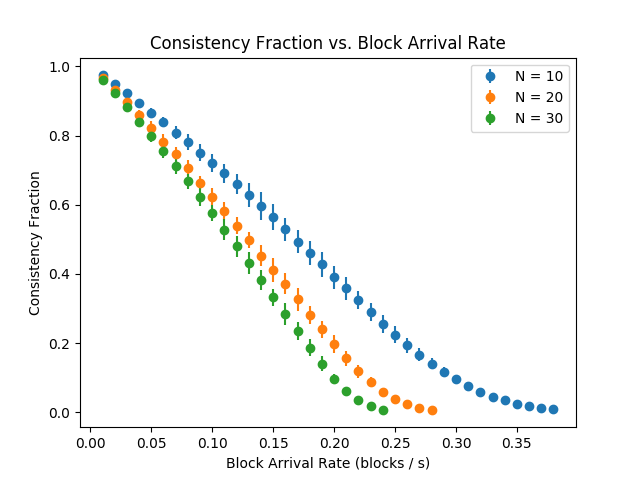}
    \caption{The consistency fraction.}
    \label{fig:FracConsistent-BlockRate}
\end{figure}
Figure \ref{fig:FracConsistent-BlockRate} captures the relationship between increasing block rates and 
the mean fraction of consistent peers.
There appears to be an inflection point when one half of peers are consistent -- above this point the graph is concave; below it is convex.

\noindent{\textbf{Growth Rate of the Distinguished Path}} --
\begin{figure}[ht!]
    \centering
    \includegraphics[width=0.65\columnwidth]{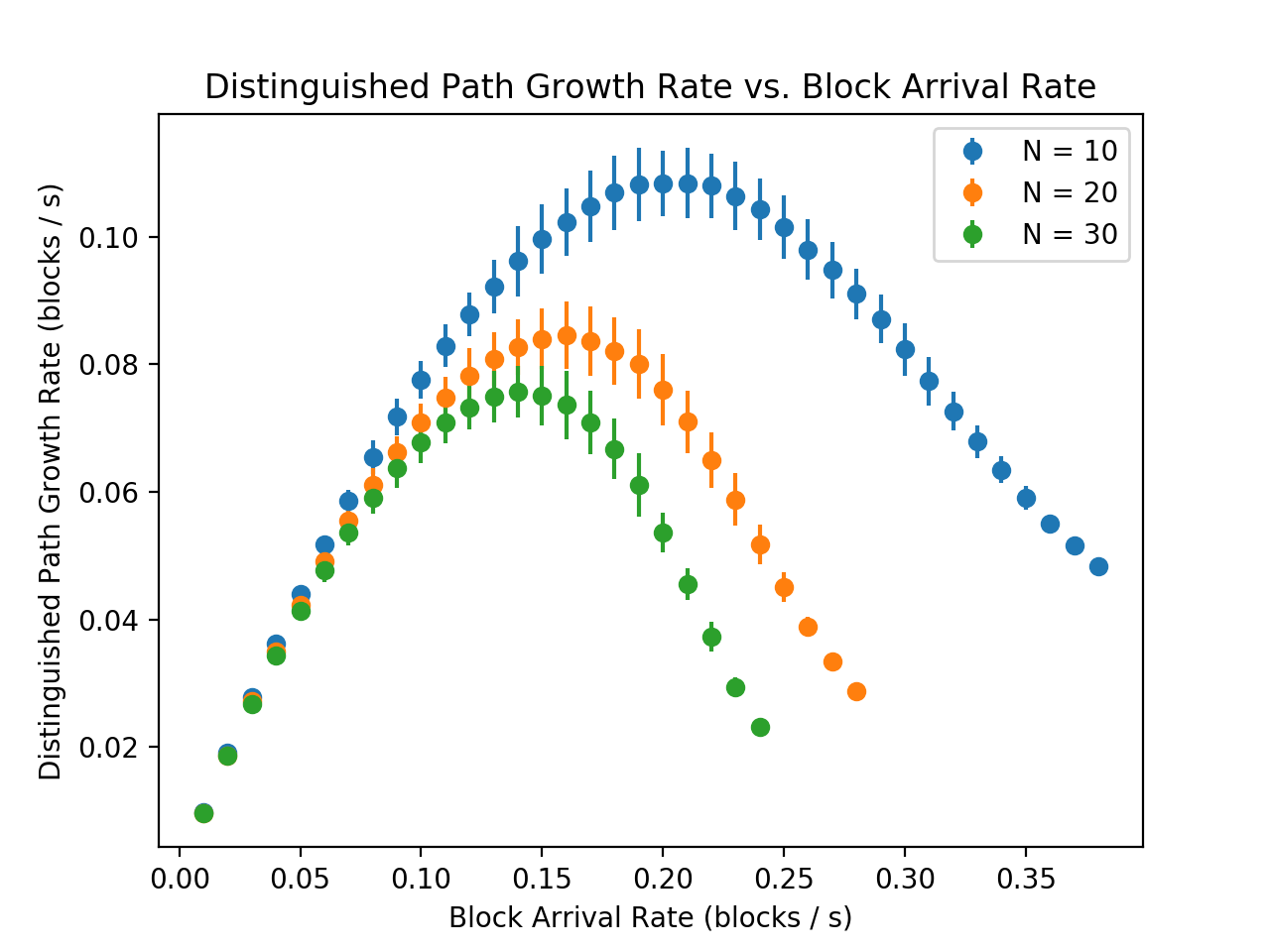}
    \caption{The growth rate of the distinguished path under the tree policy.}
    \label{fig:GrowthRate-BlockRate}
\end{figure}
Figure \ref{fig:GrowthRate-BlockRate} shows the relationship between the block arrival rate and the growth rate of the distinguished path.
As with consistency fraction, the growth rate curve appears to have an inflection point.
We observe that the growth rate of the distinguished path appears to have a unique maximum.
However, the block arrival rate that produces this maximum is not equal to the rate that minimizes the cycle length or the rate that produces an inflection point for the consistency fraction.

\noindent{\textbf{Age of Information}} --
\begin{figure}[ht!]
    \centering
    \includegraphics[width=0.65\columnwidth]{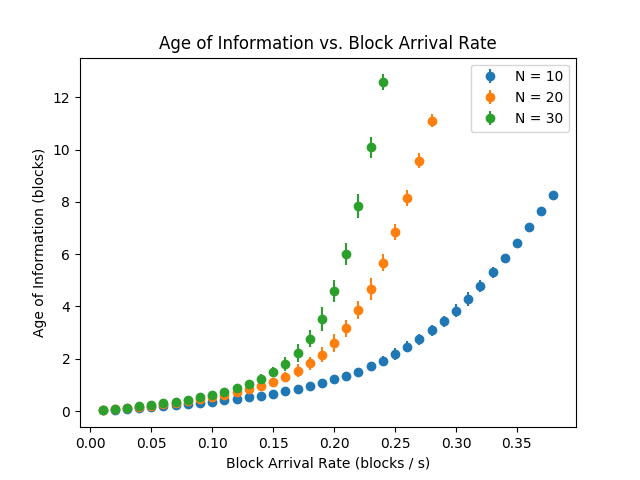}
    \caption{The average age of information.}
    \label{fig:NumBlocksBehind-BlockRate}
\end{figure}
Figure \ref{fig:NumBlocksBehind-BlockRate} depicts, for a typical peer $p$, the mean number of blocks behind consistency.
As expected, the age of information at a typical peer tends to infinity as the block arrival rate approaches the critical value.

\subsection{Real Data}
In this section, we analyze the performance of the blockchain system using real data from the Bitcoin Blockchain network \cite{blockchair2020data, nakamoto2008bitcoin}.
We do not consider the growth rate metric as the data given in \cite{blockchair2020data} does not include the blockchain DAG.

\noindent {\bf Experimental Setup} 
We consider a network of $3500$ peers connected on $4$ different network topologies -- the complete network, the (one-dimensional) torus with degree $32$, $32$-regular tree and a $32$-regular random network \cite{janson2011random}.
We choose this set of parameters guided by the measurement study in \cite{decker2013information}, which measures Bitcoin network and reports that it contains about $3500$ peers (this is an estimate), with a typical peer being connected to $32$ other peers. 
According to the measurement study \cite{gencer2018decentralization}, the average bandwidth of a peer in the Bitcoin P2P network is $73.1$ Megabits per second.
The Bitcoin protocol \cite{nakamoto2008bitcoin} specifies that a block is $1$ Megabyte ($8$ Megabits). 
Thus, we use a Poisson process of rate $\frac{73.1\text{ Mbps}}{{8\text{ bits per byte}}} = 9.14\text{ blocks/s}$ for communication among peers.
For the block arrivals, we consider the real data of block generation in the Bitcoin network provided in \cite{blockchair2020data}. 
This data set consists of $2000$ consecutive block arrival times between December $27$, $2019$ till January $09$, $2020$. 
We assume that each arriving block arrives at a peer chosen uniformly at random, as the data set does not provide complete information on who mined the block.

\noindent {\bf Results} - We provide numerical results on the metrics identified in Section \ref{sec:interpretations} in Tables \ref{tab:time_to_consistency}, \ref{tab:cycle_length}, \ref{tab:consistency_fraction} and \ref{tab:age_of_information} along with $95\%$ confidence intervals.
For comparison, we also simulate the same setup replacing the real data for block arrivals with an equivalent Poisson process of rate $1$ block per $600$ seconds.
This is the block arrival rate specified in the Bitcoin whitepaper \cite{nakamoto2008bitcoin} and can be verified by examining inter-arrival times for blocks in the data set collected in \cite{blockchair2020data}.

We observe from our results in Tables  \ref{tab:time_to_consistency}, \ref{tab:cycle_length}, \ref{tab:consistency_fraction} and \ref{tab:age_of_information} that, for the metrics we consider, the estimates on the real block arrival data and the equivalent Poisson data are similar, thereby providing another validation for the Poisson block arrival model typically used in the literature \cite{nakamoto2008bitcoin, bagaria2019prism, yang2019prism, papadis2018stochastic, li2018blockchain, li2019markov, frolkova2019bitcoin}.

\begin{table}[ht]
\caption{Time to Consistency (s)}
\centering
\begin{tabular}{| c || c | | c |}
    \hline
    Network Topology & Poisson Input & Bitcoin Input \\
    \hline
    \hline
    Complete & $1.91 \pm 0.0158$ & $1.92 \pm 0.0172$\\
    \hline
    Torus & $21.5 \pm 0.103$ & $21.5 \pm 0.121$\\
    \hline
    Tree & $48.4 \pm 0.636$ & $48.6 \pm 0.701$\\
    \hline
    Random & $1.97 \pm 0.0167$ & $1.98 \pm 0.0191$\\
    \hline
\end{tabular}
\label{tab:time_to_consistency}

\end{table}\begin{table}[ht]
\caption{Cycle Length (s)}
\centering
\begin{tabular}{| c || c | | c |}
    \hline
    Network Topology & Poisson Input & Bitcoin Input \\
    \hline
    \hline
    Complete & $625 \pm 0.000718$ & $551 \pm 0.000724$\\
    \hline
    Torus & $647 \pm 2.026$ & $570 \pm 1.99$\\
    \hline
    Tree & $674 \pm 6.73$ & $603 \pm 9.40$\\
    \hline
    Random & $625 \pm 0.000621$ & $551 \pm 0.000710$\\
    \hline
\end{tabular}
\label{tab:cycle_length}
\end{table}

\begin{table}[ht]
\centering
\caption{Consistency Fraction}
\begin{tabular}{| c || c | | c |}
    \hline
    Network Topology & Poisson Input & Bitcoin Input \\
    \hline
    \hline
    Complete & $0.998 \pm 1.95\text{e-}5$ & $0.998 \pm 2.34\text{e-}5$\\
    \hline
    Torus & $0.983 \pm 0.000154$ & $0.981 \pm 1.60e-4$\\
    \hline
    Tree & $0.977 \pm 0.000624$ & $0.974 \pm 0.000758$\\
    \hline
    Random & $0.998 \pm 1.93\text{e-}5$ & $0.998 \pm 1.88\text{e-}5$\\
    \hline
\end{tabular}
\label{tab:consistency_fraction}
\end{table}

\begin{table}[ht]
\caption{Age of Information (blocks)}
\centering
\begin{tabular}{| c || c | | c |}
    \hline
    Network Topology & Poisson Input & Bitcoin Input \\
    \hline
    \hline
    Complete & $0.00153 \pm 1.95\text{e-}5$ & $0.00174 \pm 2.32\text{e-}5$\\
    \hline
    Torus & $0.0170 \pm 0.000168$ & $0.0193 \pm 0.000177$\\
    \hline
    Tree & $0.0232 \pm 0.000662$ & $0.0261 \pm 0.000774$\\
    \hline
    Random & $0.00161 \pm 1.93\text{e-}5$ & $0.00183 \pm 1.88\text{e-}5$\\
    \hline
\end{tabular}
\label{tab:age_of_information}
\end{table}

From the results in Tables \ref{tab:time_to_consistency}, \ref{tab:cycle_length}, \ref{tab:consistency_fraction},\ref{tab:age_of_information}, it appears that the impact of the network topology on the performance of the Bitcoin network (using the data in \cite{decker2013information, gencer2018decentralization}) is minimal. 
This is in large part due to the fact that block propagation delay is much lower than the block arrival rate \cite{decker2013information}.
However, the results from our synthetic simulations suggest that with increasing network sizes (and thus with increasing block propagation delay), the effect of network topology should become more pronounced.
The current trend of blockchain adoption suggests that that the size of blockchain networks is rapidly increasing \cite{decker2013information, gencer2018decentralization}, and thus understanding the impact of network size and topology on performance is crucial for designing future blockchain systems.

\section{Related Work}
\label{sec:related_work}

We classify the related work into three categories - Peer-to-Peer networks, queueing theory, and a growing body of blockchain literature.

\noindent {\bf Peer-to-Peer Networks and Gossip Algorithms} - 
Stochastic models for standard P2P systems have been widely studied in the literature.
Yang and de Veciana~\cite{yang2004service} study the rate at which a file can be spread to peers on a P2P network under bursty requests, \textit{e.g.} when a new episode of a popular television show first becomes available.
Qiu and Srikant~\cite{qiu2004modeling} develop a steady-state fluid model to study the interactions between the number of peers who can disseminate a file and the number of peers requesting that file on a P2P network.
Massouli\'e and Vojnovic~\cite{massoulie2005coupon} use fluid models study the dynamics of peers entering and leaving a P2P network when the pieces of a file are initially distributed amongst peers.
Zhu and Hajek~\cite{zhu2012stability} study the stability of P2P systems in the context of the \emph{missing piece syndrome}, wherein the only peer with a particular piece of a file departs the system.
Baccelli \textit{et al.} \cite{baccelli2013p2p} studies P2P packet dissemination with non-uniform connectivity of the peers but do not consider network congestion.
All of these systems study P2P file transfers with dynamic P2P networks under the assumption that all pieces of a file to be disseminated exist at some peer at time 0.
Our work assumes a static P2P network, but we analyze the dissemination of new files that arrive to the system exogenously.

In the spirit of the previously mentioned literature, there is a body of research concerned with developing \emph{gossip algorithms} with the goal of solving distributed algorithmic problems, such as computing the average of measurements taken by several peers, using local algorithms.
See~\cite{shah2009gossip} for a more detailed discussion of these algorithms.
In addition, these algorithmic results concern the spreading time of a rumor on a P2P network with respect to the underlying graphical topology of the peers
~\cite{chierichetti2010rumour, fountoulakis2010rumor,  fountoulakis2012ultra, ganeshnotes,  panagiotou2017asynchronous, sanghavi2007gossiping, shah2009gossip}.
As above, these papers assume that all content (rumors) to be spread exists in the network at time 0; our work incorporates exogenous arrivals to study block propagation on P2P networks in blockchain-like systems.

Ioannidis \textit{et al.}~\cite{ioannidis2009optimal} study hybrid networks where a server updates peers in a P2P network on a real-time situation such as road traffic, and the peers share updates with each other in order to minimize the overall age of information.
While the authors of~\cite{ioannidis2009optimal} do study peer-to-peer networks in the context of information arrivals, only the newest information is of concern in their setup and thus their model does not capture the causality of references between arriving blocks to a blockchain.
Our model enforces that under stability, all peers are aware of all blocks.

\noindent {\bf Queueing Theory Approaches to Blockchain Systems} - 
Due to the temporal block arrival dynamics, recent research efforts have analyzed queueing models for blockchain systems and have focused specifically on transaction, without considering the impact of the distributed network dynamics. 
Many of these papers~\cite{li2018blockchain, li2019markov, ricci2019learning, kawase2017transaction, misic2019modeling} study the duration between when a transaction is introduced to a cryptocurrency system and when it is included in a block.

Li~\textit{et al.} \cite{li2018blockchain, li2019markov} consider a model for the blockchain where transactions arrive according to a stationary arrival process into the system. 
The model assumes that each transaction arrives into the network, waits for a random independent time duration to be included in a block, and then another random independent duration when this block is disseminated to all peers, and then exits the queueing system.
These models however, do not capture the bandwidth limitation of the P2P network. 
In contrast, in our model, the network is bandwidth-limited and block dissemination times of depend on the instantaneous network congestion.

Frolkova and Mandjes~\cite{frolkova2019bitcoin} use a $G/M/\infty$ queue with batch departures to model blockchain systems -- in their model, if a block $b$ completes its service, all blocks which arrived to the system before block $b$ which are still in service also depart the system.
The infinity-server model in~\cite{frolkova2019bitcoin} implicitly assumes unbounded communication bandwidth; our work considers block propagation on arbitrary networks of $N$ peers with bounded bandwidth.
Ricci \textit{et al.}~\cite{ricci2019learning} and Kawase and Kasahara~\cite{kawase2017transaction} use the $M/G/1$ queue and one of its variants to study the amount of time from when a transaction is created to when it is included in a block.
Despite the fact that the blockchain protocol is designed to address consensus on distributed ledgers~\cite{nakamoto2008bitcoin}, all of \cite{frolkova2019bitcoin, ricci2019learning, kawase2017transaction} analyze the blockchain as a centralized ledger.
Our results address the fundamental distributed dynamics underlying consensus in blockchain systems.

Misic \textit{et al.}~\cite{misic2019modeling} model the Bitcoin blockchain system using a Jackson network of $M/G/1$ queues.
However, they assume that the capacity of their network far exceeds the block arrival rate so that their model is \textit{de facto} stable and scalable.
Our work approaches blockchain dynamics in more generality and also bounds the stability region of the system in order to asses scalability.
Our model is the first to consider both distributed consensus dynamics as well as stability and scalability of blockchain systems.

\noindent {\bf Other Blockchain Models} - 
Papadis~\textit{et al.}~\cite{papadis2018stochastic} propose a stochastic network model for blockchain systems -- \cite{papadis2018stochastic} considers the limit of a P2P model when the communication delays are negligible compared to the block arrival rate.
Their paper provides no explicit analysis of stability.
In contrast, we model communication delays and thus establish that the system need not always be stable. 
Furthermore, we derive bounds on the stability region.
In addition, our model explicitly considers the evolution of the blockchain graph structure, allowing us to characterize the performance of various policies by which blocks add references.


Several papers in the literature make the implicit assumption that blockchain systems constructed according to the tree policy are one-ended in the temporal limit~\cite{nakamoto2008bitcoin, li2018blockchain, li2019markov, frolkova2019bitcoin, papadis2018stochastic, eyal2018majority, gobel2016bitcoin, bagaria2019prism, yang2019prism, buterin2013ethereum}.
The throughput optimal policy is introduced in~\cite{lewenberg2015inclusive}, which also implicitly assumes one-endedness.
This paper provides conditions for when these assertions hold.
Pass \textit{et al.}~\cite{pass2017analysis} and Sompolinsky and Zohar~\cite{sompolinsky2015secure} show a condition equivalent to one-endedness under the tree policy, assuming unbounded bandwidth.
Thus, their analyses do not shed insight on the effects of bandwidth limits and communication delays, and network topologies.

\noindent {\bf Measurements and System Implementations} - Decker and Wattenhofer \cite{decker2013information} perform a measurement analysis of the Bitcoin P2P network.
They provide measurements on the number of peers and average degree in the network.
They also find that information propagation in the Bitcoin P2P network resembles a gossip protocol.
Recently, there have been efforts to modify the original Bitcoin protocol \cite{nakamoto2008bitcoin}, to improve the blockchain system under various metrics. 
For example, Bagaria \textit{et al.}~\cite{bagaria2019prism} and Yang \textit{et al.}~\cite{yang2019prism} propose improvements to the throughput of blockchain-like systems from an information theoretic perspective over the standard models of Bitcoin and Ethereum, the two largest blockchain implementations~\cite{nakamoto2008bitcoin, buterin2013ethereum}.
Bojja \textit{et al.} and Fanti \textit{et al.} \cite{bojja2017dandelion, fanti2018dandelion++} propose new protocols for P2P communication that preserve peer anonymity for blockchain systems. 
Conducting stability and scalability analyses for these protocols is an exciting avenue for future work.

\section{Concluding Remarks}
\label{sec:conclusion}
In this paper, we model blockchain systems as a gossiping protocol on a peer-to-peer network subject to exogenous block arrivals.
We show that when the gossiping protocol is stable, any blockchain constructed according to the tree or throughput-optimal policy is one-ended.
We then determine bounds on the maximum block arrival rate for a P2P network $H$ such that the stochastic model is stable.
Following this analysis, we examine the scalability of several commonly studied network topologies.
We then verify our insights through simulations on both synthetic and real data.

There are several open problems that arise from this paper.
Future improvements to our bounds in Theorem \ref{thm:stability} would allow for more complete scalability analyses.
This may require the development of novel mathematical tools.
In addition, having analytic expressions for the performance metrics identified in this paper is important for assessing and comparing different design choices for the network.

In this paper, we study the fundamental aspects of distributed consensus in blockchain systems, namely the dynamics of the blockchain DAG and the requisite stability and scalability.
Extending our model to use transactions as the atomic unit is a natural direction for future work.
\\
\newline
\noindent {\bf Acknowledgements} 
This work was completed while AG and AS was at The University of Texas at Austin.
AG was supported by a Ripple Foundation Fellowship, awarded to The University of Texas at Austin.
AS thanks Fran\c cois Baccelli for support and funding through the Simons Foundation grant (\#197892) awarded to The University of Texas at Austin. AS also thanks Sergey Foss for pointing out reference \cite{foss2004overview}.
The authors thank anonymous reviewers and Daniel Sadoc Menasch\'e for their insightful comments on the presentation of our results.

\pagebreak
\bibliographystyle{plain}
\bibliography{references}

\pagebreak
\appendix
\section{Some Technical Considerations for Limiting DAGs}
\label{app:technical-dag}
In this appendix, we discuss the notion of convergence of a sequence of finite DAGs $(G_k)_{k \in \mathbb{N}}$ to a limit which is an infinite DAG.
This allows us to define a limiting DAG $G(\infty)$ (which is the temporal limit of the blockchain DAG) and characterize the confirmed blocks contained in the limiting DAG.
Such problems of convergence are studied in \cite{aldous2007processes, benjamini2011recurrence, baccelli2019doeblin}.

We use the space $\mathcal{B}_*$ of locally finite, connected, and rooted DAGs where each vertex has a unique mark as defined in \cite{aldous2007processes}.
We consider the metric space $(\mathcal{B}_*, d_*)$ where $d_*(G_1, G_2) = \frac{1}{r+1}$ if $r$ is the least non-negative integer such that the $r$-balls centered at the roots of $G_1$ and $G_2$ are equal.
If $G_1, G_2$ are finite DAGs in $\mathcal{B}_*$ with $G_1 = G_2$, we use the convention that $d_*(G_1, G_2) = 0$ as for any large enough radius $r$, the $r$-ball encompasses the entire (finite) DAG.
It is established in \cite{aldous2007processes} that $(\mathcal{B}_*, d_*)$ is a complete metric space.

\begin{lemma}
  Any sequence of rooted DAGs $(G_k)_{k \in \mathbb{N}}$ satisfying the following conditions is Cauchy in $(\mathcal{B}_*, d_*)$:
  \begin{itemize}
      \item Every DAG in $(G_k)_k$ has the same root vertex.
      \item $G_k$ is connected for all $k \in \mathbb{N}$.
      \item $G_k \subseteq G_{k+1}$ for all $k \in \mathbb{N}$.
      \item There exists $n \in \mathbb{N}$ such that no vertex in $G_k$ has degree greater than $n$, for all $k \in \mathbb{N}$.
  \end{itemize}
  Furthermore, the limit $\lim_{k\to\infty}G_k = \bigcup_{k \in \mathbb{N}}G_k$.
  \label{lem:dag-cauchy}
\end{lemma}
\begin{proof}
    For all $k \in \mathbb{N}$, denote by $V_{k}^{r}$ the vertices of $G_k$ within $r$ hops of the root.
    It is clear that for each $r \in \mathbb{N}$, there is a finite index $k_r$ such that for $m \geq k_r$, $V_{m}^{r} = V_{k_r}^{r}$, for otherwise there exists a finite index $n_r$ such that some vertex in the $r$-ball of the root of $G_{n_r}$ has degree greater than $n$.
    
    In particular, for any two indices $m, l \geq k_r$, $d_*(G_m, G_l) \leq \frac{1}{r+1}$ as the $r$-balls centered at their common root vertices agree.
    It follows immediately that the sequence $(G_k)_{k \in \mathbb{N}}$ is Cauchy in $(\mathcal{B}_*, d_*)$. 
\end{proof}

The first three conditions imposed in Lemma \ref{lem:dag-cauchy} are immediate from the construction of our system model in Section \ref{sec:model} by considering block 0 as the common root vertex and marking each block with its arrival index.
The fourth condition is established for the tree and throughput-optimal policies in Lemma \ref{lem:locally-finite}.

Note by choosing $G_k$ to be $G(A_k)$, which is the blockchain at the moment of the $k$-th block arrival, we have a sequence of graphs satisfying the conditions in Lemma \ref{lem:dag-cauchy}, and thus $G(\infty) = \lim_{k\to\infty}G(A_k) = \bigcup_{k \in \mathbb{N}}G(A_k)$.

\section{Proof of Proposition \ref{prop:maximal_path}}
\begin{proof}[Proof of Proposition \ref{prop:maximal_path}]
\label{prf:maximal_path}
   We proceed by strong induction.
   For the base case, note that the edge from a vertex 1 must connect to vertex 0 by the definition of an edge selection policy.
   
   Suppose next that all maximal paths on block structure consisting of vertices $0, \ldots, k$ end at vertex 0 for some $k \in \mathbb{N}$, and suppose that vertex $k+1$ arrives at time $t_{k+1}$.
   Every edge from vertex $k+1$ either connects directly to vertex 0, in which case we have a maximal path ending at vertex 0, or ends at some vertex in $\{1, \ldots, k\}$.
   In the latter case, all maximal paths from vertex $k+1$ include as a subpath the maximal path from some vertex $i \in \{1, \ldots, k\}$; thus from the definition of a maximal path, all maximal paths from vertex $k+1$ end at vertex 0 at time $t_{k+1}$, and as the edges are fixed henceforth, the proposition follows.
\end{proof}

\section{Proofs from Section \ref{sec:structural-properties}}
\subsection{Proof of Proposition \ref{prop:trust}}
\label{prf:trust}
\begin{proof}[Proof of Proposition \ref{prop:trust}]
  Since $b$ is a confirmed block, all but finitely many blocks of index greater than $b$ have a directed path to $b$ in $G(\infty)$.
  Since the arrival rate at each peer $p \in \{1,\cdots,N\}$ is positive, all peers eventually add infinitely many blocks with index greater than $b$.
  As only finitely many of those blocks cannot have a path to $b$ (as $b$ is confirmed), it means that there is eventually a block added by peer $p$ at some $t$, with a path to $b$ in $G_p(t)$, and thus also a path in $G(\infty)$.
\end{proof}

\subsection{Proof of Lemma \ref{lem:infinitely-confirmed}}
\label{prf:infinitely-confirmed}
\begin{proof}[Proof of Lemma \ref{lem:infinitely-confirmed}]
  As $G(\infty)$ is locally finite, it contains an infinite path, and in particular each confirmed block lies on some infinite path (for otherwise a confirmed block must have infinite in-degree).
  We first show that there exists an infinite path in $\widehat{G}(\infty)$.
  Suppose otherwise, and all connected components of $\widehat{G}(\infty)$ have finite cardinality. Thus, $\widehat{G}(\infty)$ is an union of infinite collection of finite non-empty connected DAGs. 
  Thus, each block in a connected component of $\widehat{G}(\infty)$ (which is confirmed by construction of $\widehat{G}(\infty)$), has infinitely many blocks that do not reference it, which contradicts the definition of a block being confirmed.
  Thus, $\widehat{G}(\infty)$ contains at-least one infinite connected component, i.e., there is at-least one infinite path $p_1$.
  Suppose there are two infinite paths in $\widehat{G}(\infty)$ that intersect only finitely often. 
  This implies that there are confirmed blocks on either paths of $\widehat{G}(\infty)$, that are missing references from infinitely many other blocks.
  As in a locally finite DAG $G(\infty)$, all neighbors of a confirmed block are also confirmed, the two paths contradicts the definition that blocks in $\widehat{G}(\infty)$ are confirmed.
\end{proof}

\subsection{Proof of Lemma \ref{lem:one-ended-confirmed}}
\label{prf:one-ended-confirmed}
\begin{proof}[Proof of Lemma \ref{lem:one-ended-confirmed}]
  We proceed by contradiction.
  Suppose $G(\infty)$ is one ended and the number of confirmed blocks are finite. 
  This implies that there is a confirmed block of greatest index, denoted by $b^{'}$.
  Note by definition of a confirmed block, there are infinitely many blocks having a path to $b^{'}$. 
  As the DAG is locally finite, this implies that block $b^{'}$ lies on an infinite path denoted by $p$.
  Denote by block $b$ the first block with index greater than $b^{'}$ which lies on the infinite path $p$.
  Let $(v_k)_{k \geq 1}$, be the collection of all blocks indexed greater than or equal to $b+1$, with no directed path to $b$ in $G(\infty)$.
  The existence of such an infinite sequence $(v_k)_{k \geq 1}$ follows as block $b$ is not confirmed.
  It suffices to now establish that there exists an infinite path $(v_{k_i})_{i \geq 1}$ as a subset of $(v_k)_{k \geq 1}$.
  The existence of such a infinite path contradicts the fact that $G(\infty)$ is one-ended, as the infinite path $(v_{k_i})_{i \geq 1}$ and the infinite path $p$, intersect only finitely many times.
  
  By construction, all blocks in $G(\infty)$ have a path to $0$.
  Suppose that there is no infinite path in $(v_k)_{k \geq 1}$.
  This implies that the maximum DAG distance (number of ``hops'') from block $0$, to any block in $(v_k)_{k \geq 1}$ is finite.
  However, this contradicts the local finiteness of $G(\infty)$. 
  Thus there exists an infinite path  $(v_{k_i})_{i \geq 1}$ as a subset of $(v_k)_{k \geq 1}$.
\end{proof}

\section{Proofs from Section \ref{sec:one-ended-policies}}
\subsection{Proof of Lemma \ref{lem:locally-finite}}
\label{prf:locally-finite}
\begin{proof}[Proof of Lemma \ref{lem:locally-finite}]
  Notice that when a block arrives at time $t$, the set of outgoing edges are the subset of blocks that have arrived before it.
  Almost surely, for all $t$, only finitely many block have arrived before time $t$.
  Thus, almost-surely, all blocks have finite out-degree.
  
  In both the tree and throughput-optimal policies, new blocks are only connected to leaves. 
  Thus, each peer can add at most one incoming edge to any block, since that block is no longer a leaf after the addition of such an edge.
  Thus, the in-degree of all blocks in $G(t)$, under both policies is bounded above by $N$.
\end{proof}

\subsection{Proof of Theorem \ref{thm:tree-one-ended}}
\label{prf:tree-one-ended}
\begin{proof}[Proof of Theorem \ref{thm:tree-one-ended}]
  Consider the set of blocks $(v_k)_{k \in \mathbb{N}}$, such that for each $k \geq 1$, $v_k$ corresponds to the start of the distinguished path in $G(C_k)$. 
  First, we shall establish under the hypothesis of the theorem, that for all $k \geq 0$, $v_k \neq v_{k+1}$. 
  To do so, fix any $k \geq 0$, and denote by $Z_k$ to be the length of the distinguished path in $G(C_k)$. Denote by the first exogenous block $i \in \mathbb{N}$, to arrive at a peer $p \in \{1,\ldots,N\}$, at time $C_k < t <C_{k+1}$. 
  By the definition of tree policy, this extends the length of the distinguished path in $G_p(t)$ to $Z_k + 1$. 
  However, at time $C_{k+1}$, all peers are aware of the block $i$, all peers' distinguished path length is at-least $Z_{k}+1$ and thus, the length of the distinguished path in $G(C_{k+1})$ is of length at least $Z_k + 1$. 

  \begin{lemma}
    Suppose all peers use the tree policy.
    Let $C$ be a time of consistency.
    Let $v_C$ be the start of the distinguished path in $G(C)$. 
    Then, for all $t \geq C$, the distinguished path in $G(t)$ passes through $v_C$.
    \label{lem:distinguished_path}
  \end{lemma}
  \begin{proof}
    We proceed by induction. 
    Denote by times $\mathcal{E}_1,\mathcal{E}_2,\ldots$ the set of times after $C$ when either a new block arrives exogenously, or a communication occurs between peers. 
    Notice that event time $\mathcal{E}_1$ always corresponds to an exogenous arrival, as all peers have all the blocks that have arrived in the network thus far at time $C$. 
    At event time $\mathcal{E}_1$, the outgoing edge from the newly arrived block must point to $v_C$, no matter the peer at which it arrives. 
    This follows from the stipulation that the peers follow the tree policy to connect blocks and from the choice of $v_C$. 
    Thus, the distinguished path in $G(\mathcal{E}_1)$ passes through $v_C$. 
    Now assume as the induction hypothesis that for some $l \geq 1$, i.e., after all event time $\mathcal{E}_l$, at all peers $p \in \{1,\ldots,N\}$, the distinguished path in $G_p(\mathcal{E}_l)$ passes through $v_C$. 
    Consider event time $\mathcal{E}_{l+1}$. 
    If it is an exogenous arrival at some peer $p \in \{1,\ldots,N\}$, then by the tree policy, this block will connect to the start of the distinguished path in $G_p(\mathcal{E}_l)$, which, from the induction hypothesis, passes through $v_C$. 
    Suppose time $\mathcal{E}_{l+1}$ corresponds to a communication event, at which some peer $p \in \{1,\ldots,N\}$, receives a block $i \in \mathbb{N}$ that arrives exogenously at or before time $\mathcal{E}_l$. 
    Note that at time $C$, all peers were aware of the same set of blocks, thus the communication event at time $\mathcal{E}_{l+1}$, must correspond to an exogenous arrival to some peer $q \in \{1,\ldots,N\}\setminus \{p\}$ at time $\mathcal{E}_{T_q}$, for some $T_q \in \{1,\ldots,l\}$. 
    The maximal path from block $i$ in $G_q(\mathcal{E}_{T_q})$, passes through $v_C$ as a result of the induction hypothesis. 
    Now, at time $\mathcal{E}_{l+1}$, after peer $p$ becomes aware of block $i$, either the distinguished path remains unchanged from time $\mathcal{E}_l$, or the distinguished path starts from the newly arrived block $i$. 
    In the former case, the induction hypothesis (applied to peer $p$ at time $\mathcal{E}_l$) implies that the distinguished path goes through $v_C$. 
    In the latter case also, we show that the maximal path from $i$ passes through $v_C$ and ends at $0$. 
    To see this, observe that at time $\mathcal{E}_{T_q}$ in DAG $G_q(\mathcal{T}_q)$, the maximal path from $i$ passes through $v_C$ and ends at $0$ (which follows from the induction hypothesis applied to peer $q$ at time $\mathcal{E}_{T_q}$). 
    But, as the outgoing edges are fixed for all blocks upon their arrivals, the maximal path from $i$ in $G_p(\mathcal{E}_{l+1})$ either - {\em (i)} goes through $v_C$ until $0$ as the maximal path from $v_C$ in $G_p(C)$ (and thus in $G_p(t)$ for all $t \geq C)$) ends at $0$, {\em (ii)} or is disconnected from $v_C$ and hence from $0$. 
    However, the latter cannot occur as block $i$ is the start of the distinguished path in $G_p(\mathcal{E}_{l+1})$, which by definition must end at $0$.
  \end{proof}

  For every $t \geq 0$, let $\widetilde{G}(t)$, be the version of $G(t)$ with its edges reversed. 
  Define $\widetilde{G}(\infty) := \cup_{t \geq 0} \widetilde{G}(t)$. 

  \begin{lemma}
    Any infinite ray in $\widetilde{G}(\infty)$ must pass through all but finitely many $(v_k)_{k \in \mathbb{N}}$. 
    In particular, $\widetilde{G}(\infty)$ is one-ended.
    \label{lem:one_ended}
  \end{lemma}
  \begin{proof}
    Fix any (infinite) ray in $\widetilde{G}(\infty)$, starting at any block (vertex) $i \in \mathbb{N}$. 
    We argue by contradiction. 
    Suppose there are only finitely many $v_{k_1},\ldots v_{k_m}$ through which a ray from block $i$ passes.
    Consider time $C_{k_m +1}$. 
    There exists at least one block $j_m$ (arriving at time $t_m \geq C_{k_m +1}$), that lies on the infinite ray, and arrives after time $C_{k_m +1}$. This follows as almost surely, only finitely many blocks have arrived before time $C_{k_m +1}$ and the ray contains infintely many blocks. 
    From the construction of $(v_k)_{k \in \mathbb{N}}$, the maximal path from block $j_m$ in $G(t_m)$ passes through $v_{k_m+1}$ and also through $i$, since a path from $i$ to $j$ exists in the reversed DAG $\widetilde{G}(\infty)$. 
    As the maximal path from $j$ in $G(\infty)$ is unique and passes through $v_{k_m +1}$ and $i$, any path from $i$ to $j$ must pass through $v_{k_m+1}$.
    This contradicts the fact that the ray starting at $i$ and passing through $j$ does not pass through $v_{k_m+1}$.

    In this lemma we show that there is an infinite sequence $(v_k)_k$ of vertices in $\widetilde{G}(\infty)$ such that any infinite ray passes through each $v_k$.
    It follows that any two infinite rays intersect at each $v_k$.
    Then for any two rays $p_1, p_2 \in \widetilde{G}(\infty)$, one can choose $p_3 = p_1$ in Definition \ref{def:one-ended}, establishing that all infinite rays in $\widetilde{G}(\infty)$ are equivalent; hence $\widetilde{G}(\infty)$ is one-ended.
  \end{proof}

  As there is a bijection from rays in $\widetilde{G}(\infty)$ to rays in $G(\infty)$, it follows that $G(\infty)$ is one-ended as claimed.
\end{proof}

\subsection{Proof of Corollary \ref{cor:markov_dynamics}}
\label{prf:markov-dynamics}

\begin{proof}[Proof of Corollary \ref{cor:markov_dynamics}]
  The result follows from the fact that any positive recurrent Markov chain returns to each of its states within finite time, and the fact that the process is assumed to evolve from an initial condition wherein all peers $i$ have identical block sets $B_i(t)$ (namely only the blocks in $G_0$).
\end{proof}

\subsection{Proof of Corollary \ref{cor:distinguished_path}}
\label{prf:distinguished-path}
\begin{proof}[Proof of Corollary \ref{cor:distinguished_path}]
  Both directions follow from Lemma \ref{lem:distinguished_path} as follows.
  
  Let there exists such a time of consistency $C$ such that $b$ is on the distinguished path in $G(C)$.
  From Lemma \ref{lem:distinguished_path}, it follows that all blocks arriving to the system after $C$ have a path to $b$ since $b$ is on the distinguished path in $G(t)$ for all $t \geq C$.
  As only finitely many blocks arrive to the system after the arrival of $b$ and before $C$, it follows that $b$ is confirmed.
  
  Assume that for every time of consistency $C$, the block $b$ is not on the distinguished path in $G(C)$.
  From Lemma \ref{lem:distinguished_path}, all blocks arriving to the system after time $C$, only have directed paths to blocks on the distinguished path in $G(C)$.
  This is because the tree policy adds exactly one outgoing edge from each arriving block.
  Thus in this case, there are infinitely many blocks which do not have a path to $b$; hence $b$ is not a confirmed block.
\end{proof}

\subsection{Proof of Theorem \ref{thm:throughput-optimal-one-ended}}
\label{prf:throughput-optimal-one-ended}
We establish the following Lemmas before proving Theorem \ref{thm:throughput-optimal-one-ended}.
\begin{lemma}
  Let $C$ be the last time of consistency before the arrival of a block $b$ at time $t_b$ and at some peer $p$.
  Then there is a path in the DAG $G(t)$ for all $t \geq C_k$, from vertex $b$ to every other vertex in $G(C)$. 
  \label{lem:throughput-optimal-one-ended}
\end{lemma}
\begin{proof}
  We need only establish that all vertices in $G(C)$ are contained in a maximal path from $b$ in $G_p(t_b)$.
  
  Block $b$ has edges to all leaves in $G_p(t_b)$ from the definition of the throughput-optimal policy.
  Note that $G(C) \subseteq G_p(t_b^-)$, where $t_b^-$ is a moment of time instantaneously before $t_b$.
  As edges are fixed at the time of arrival for each block, every block $v$ in $G_p(t_b^-)$ is either a leaf or there is a path from some leaf to $v$ in $G_p(t_b^-)$.
  Thus, there is a path from $b$ to every vertex in $G(C)$ in the DAG $G_p(t_b)$.
\end{proof}

\begin{lemma}
  Consider a stable blockchain system using the throughput-optimal policy in the limit as $t \to \infty$.
  Let $(C_k)_{k\in\mathbb{N}}$ is a sequence of times of consistency such that $G(C_j) \neq G(C_k)$ if $j \neq k$.
  Every infinite path contains a block $b_k$ which arrives to the system in the time interval $[C_k, C_{k+1}]$ for all $k \in \mathbb{N}$.
  \label{lem:throughput-optimal-path}
\end{lemma}
\begin{proof}
  We proceed by contradiction. 
  Assume that there is a path from $0$ in $G(\infty)$ to a block $b$, that arrives into the system strictly after time $C_{k+1}$ and that this path contains no block arriving in the time interval $[C_k,C_{k+1}]$.
  Without loss of generality, suppose there is a reference from $b$ to some block $b'$ such that the block $b'$ arrives to the system before time $C_k$.
  This implies that at the time of arrival of block $b$, the block $b'$ is a leaf in the DAG $G_p(t)$ for some peer $p$ and some time $t > C_{k+1}$.
  Since edges are fixed upon the arrival of blocks, this further implies that $b'$ is a leaf in $G(C)$.
  Under the throughput-optimal policy, as at least one block arrives to the system in the time interval $[C_k, C_{k+1}]$, $b'$ cannot be a leaf in $G_p(t)$ for any peer $p$ and any time $t \geq C_{k+1}$.
  This contradicts the existence of such a path.
\end{proof}

\begin{proof}[Proof of Theorem \ref{thm:throughput-optimal-one-ended}]
  Recall that from Definitions \ref{def:consistent} and \ref{def:stability} there exists an infinite sequence of times of consistency $(C_k)_{k \in \mathbb{N}}$ such that $G(C_j) \neq G(C_k)$ if $j \neq k$.
  As before, let $\widetilde{G}(t)$ be the version of $G(t)$ with its edges reversed and $\widetilde{G}(\infty) := \cup_{t \geq 0}\widetilde{G}(t)$.
  
  Lemma \ref{lem:throughput-optimal-path} implies that every infinite path $p$ in $\widetilde{G}(\infty)$, contains at least one vertex that arrives in the time interval $[C_k,C_{k+1}]$, for all $k \in \mathbb{N}$.
  Without loss of generality, consider two paths $p_1$ and $p_2$ in $\widetilde{G}(\infty)$, both beginning at block 0.
  Consider the subsequences of blocks on these paths $(b_k^i)_{k \in \mathbb{N}}$ for $i \in \{1, 2\}$, such that for all $k \in \mathbb{N}$ and $i \in \{1,2\}$, the block $b_k^i$ arrived in the time interval $[C_k,C_{k+1}]$. 
  Lemma \ref{lem:throughput-optimal-path}, gives the existence of such a subsequence of any infinite path.
  It follows from Lemma \ref{lem:throughput-optimal-one-ended} that there exists directed paths from $b_k^1$ to $b_{k+1}^2$, and a path from $b_k^2$ to $b_{k+1}^1$ for all $k \in \mathbb{N}$ in $\widetilde{G}(\infty)$.
  This is because the time of consistency $C_{k+1}$ is between the arrival times of blocks $b_k^1$ and $b_{k+1}^2$.
  A similar argument gives the existence of such a path between the pair of blocks $b_{k}^2$ and $b_{k+1}^1$.
  The following path $p_3 = b_k^1 \rightarrow b_{k+1}^2 \rightarrow b_{k+2}^1 \rightarrow b_{k+3}^2 \rightarrow \ldots $ intersects both $p_1$ and $p_2$ infinitely often.
  Thus, $\widetilde{G}(\infty)$ (and hence $G(\infty)$) is one-ended.
\end{proof}

\subsection{Proof of Corollary \ref{cor:throughput-markov}}
\label{prf:throughput-markov}
\begin{proof}[Proof of Corollary \ref{cor:throughput-markov}]
  The proof is identical to that of Corollary \ref{cor:markov_dynamics}.
\end{proof}

\subsection{Proof of Corollary \ref{cor:throughput-all-confirmed}}
\label{prf:throughput-all-confirmed}
\begin{proof}[Proof of Corollary \ref{cor:throughput-all-confirmed}]
  This follows immediately from Lemma \ref{lem:throughput-optimal-one-ended} because stability implies the existence of an infinite sequence times of consistency $(C_k)_{k \in \mathbb{N}}$ such that $G(C_j) \neq G(C_k)$ if $j \neq k$.
\end{proof}

\section{Proofs from Section \ref{sec:stability}}
\subsection{Proof of Lemma \ref{lem:monotone_separable}}
\label{prf:monotone_separable}
We prove Lemma \ref{lem:monotone_separable} by separately stating and proving the following propositions.

\begin{proposition}
  For all $n \geq m,$ $X_{[m, n]}$ is causal; namely $$X_{[m, n]}(A) \geq A_n.$$
  \label{prop:causality}
\end{proposition}
\begin{proof}
  This follows as no peer can communicate the $n$-th block until after it arrives (which is at time $A_n$).
\end{proof}
\begin{proposition}
  For all $n \geq m,$ $X_{[m, n]}$ is externally monotonic; namely $$X_{[m, n]}(A') \geq X_{[m, n]}(A)$$ if $A'$ is a point process such that $A'_m \geq A_m$ for all $m \in \mathbb{N}$.
  \label{prop:external-monotonicity}
\end{proposition}
\begin{proof}

  We construct the marked point process $N$, such that the points of $N$ are the union of the points of the arrival process $A$ and the communication processes $(T_p)_p$.
  In addition, we consider a point process $N'$, such that $N'_m > N_m$ for all $m \in \mathbb{N}$.
  For the $k$-th arrival, denote by $D_k$ and $D'_k$ the departure times relative to the point processes $N$ and $N'$.
  It is clear that $D_k \leq D'_k$.
  
  External monotonicity is then established from the fact that if $N'$ is equal to $N$ except at a single point corresponding to some arrival $A_k$, there is no arrival $l$ such that $D'_l < D_l$.
  
\end{proof}
\begin{proposition}
  For all $n \geq m,$ $X_{[m, n]}$ is homogeneous; namely $$X_{[m, n]}(A + c) = X_{[m, n]}(A) + c \ \forall c \in \mathbb{R}.$$
  \label{prop:homogeneity}
\end{proposition}
\begin{proof}
  As the gossiping processes are FCFS, let $\tau$ be the time from the arrival of $E_n$ until consistency.
  Note that $\tau$ is time-invariant as it is derived from the mark of $E_n$.
  Then $X_{[m, n]}(E) + c = (E_n + \tau) + c = (E_n + c) + \tau = X_{[m, n]}(E + c)$.
\end{proof}
\begin{proposition}
  For all $n \geq m,$ $X_{[m, n]}$ is separable; namely $$X_{[m, n]}(A) = X_{[l+1, n]}(A)$$ if $X_{[m, l]} \leq A_{l+1}$.
  \label{prop:separability}
\end{proposition}
\begin{proof}
  By assumption block $E_{l+1}$ arrives at an empty system; thus block $E_n$ does not wait for the dispersal of any of the blocks $E_m, E_{m+1}, \ldots E_l$ due to the FCFS nature of the gossiping process.
  It follows that $X_{[m, n]}(E) = X_{[l+1, n]}(E)$.
\end{proof}

\begin{proof}[Proof of Lemma \ref{lem:monotone_separable}]
  A monotone separable system is one that satisfies the conditions of causality, external monotonicity, homogeneity, and separability \cite{baccelli1995saturation}.
  These are established in Propositions \ref{prop:causality}, \ref{prop:external-monotonicity}, \ref{prop:homogeneity}, and  \ref{prop:separability}.
\end{proof}

\subsection{Proof of Theorem \ref{thm:stability}}
\label{prf:stability}
\begin{proof}[Proof of Theorem \ref{thm:stability}]
    For all $n \geq m,$ $X_{[m, n]}$ is monotone separable from Lemma \ref{lem:monotone_separable}.
    As such, we proceed by providing lower bounds and upper bounds for $\mathbb{E}[X_n]$ and taking the limit of $\frac{\mathbb{E}[X_n]}{n}$ as $n\rightarrow\infty$.
    
    Ganesh \cite{ganeshnotes} gives an upper bound for $\mathbb{E}[X_1] \leq \frac{2\log N}{\phi_H}$.
    By monotone separability, we shift the arrivals so that each arrival occurs at the exact instant the previous arrival is known to all peers.
    Thus, by the strong law of large numbers, we have:
    \begin{equation}
        \mu^{-1} = \lim_{n\rightarrow\infty}\frac{X_n}{n} \leq \frac{2\log N}{\phi_H} \ a.s.
    \end{equation}
    and so a lower bound on the critical arrival rate is
    \begin{equation}
        \frac{\phi_H}{2\log N} \leq \mu.
    \end{equation}
    
    Let $S \subset \{1, \ldots, N\}$ be a non-empty proper subset and denote by $S^C$ the complement of S.
    We note that the clearing time $X_n$ for $n$ blocks is at least as long as the time it takes for all the blocks initially at peers in $S$ to be communicated to peers in $S^C$.
    In particular, if $k$ of the $n$ blocks arrive to peers in $S$, then $X_n$ is at least as long as the first $k$ attempted transmissions from peers in $S$ to peers in $S^C$.
    Shifting the arrivals so that they all occur at time 0, we have:
    \begin{equation}
        \mathbb{E}[X_n] \geq \frac{n|S|}{N}\frac{1}{\sum_{p \in S, q \in S^C}\frac{1}{d(p)}\mathbf{1}_{pq}}
        \geq \frac{n}{|S^C|}\frac{1}{\phi_H^{(S)}}.
    \end{equation}
    
    Recall from Definition \ref{def:conductance} that $\mathbf{1}_{pq}$ is the indicator for an edge between $p$ and $q$ in $H$.
    Thus, almost surely,
    \begin{align}
        \mu^{-1} = \lim_{n\rightarrow\infty} \frac{X_n}{n} = \lim_{n \rightarrow \infty} \frac{\mathbb{E}[X_n]}{n} \geq \frac{1}{|S^C|\phi_H^{(S)}}.
        \label{eqn:mu-lower-bound}
    \end{align}
    In particular, we have $$\mu^{-1} = \lim_{n\to\infty}\frac{\mathbb{E}[X_n]}{n} \geq \sup_{S \subset H}\frac{1}{|S^C|\phi_H^{(S)}}.$$
    
    Recall that for a monotone separable system, the existence of the almost sure limit $\lim_{n \rightarrow \infty}\frac{X_n}{n}$ and the almost sure equality  $\lim_{n\rightarrow\infty} \frac{X_n}{n} = \lim_{n \rightarrow \infty} \frac{\mathbb{E}[X_n]}{n} $ is given in \cite{baccelli1995saturation}.
    
    Thus, it follows that an upper bound on the critical vertex arrival rate is $\mu \leq \inf_{S \subset H}|S^C|\phi_H^{(S)}$.
\end{proof}

\end{document}